\documentclass[11pt]{article}

\usepackage[T1]{fontenc}
\usepackage[utf8]{inputenc}
\usepackage[margin=1in]{geometry}
\usepackage{amsmath,amssymb,amsthm}
\usepackage{bm}
\usepackage{booktabs}
\usepackage{graphicx}
\usepackage{microtype}
\usepackage{enumitem}
\usepackage{etoolbox}
\usepackage{xcolor}
\usepackage{tikz}
\usetikzlibrary{positioning,calc,shapes.geometric}
\usepackage{natbib}
\usepackage[colorlinks=true,linkcolor=blue,citecolor=blue,urlcolor=blue]{hyperref}

\setcitestyle{authoryear,open={(},close={)}}

\newcommand{\Xomit}[1]{}
\newcommand{\graph}{\mathcal{G}}
\newcommand{\oddcycles}{\nu}
\newcommand{\typeset}{T}
\newtoggle{version}
\toggletrue{version}
\newtoggle{comments}
\toggletrue{comments}
\newcommand{\thanh}[1]{{\color{blue}#1}}

\newcommand{\indep}{\nu'}

\newcommand{\N}{N}
\newcommand{\Ex}{\mathcal{E}}
\newcommand{\G}{\mathcal{G}}
\newcommand{\GG}{\mathcal{G}^{+V^0}}

\newtheorem{theorem}{Theorem}
\newtheorem{lemma}{Lemma}

\newtheorem{definition}{Definition}
\newtheorem{corollary}{Corollary}
\newtheorem{example}{Example}

\newtheorem{remark}{Remark}

\title{Core-Stable Kidney Exchange via Altruistic Donors}
\author{
Gergely Cs\'aji\\
HUN-REN Institute for Computer Science and Control, Budapest, Hungary\\
\texttt{csajigergely@sztaki.hu}
\and
Th\`anh Nguyen\\
Purdue University, West Lafayette, Indiana, USA\\
\texttt{nguye161@purdue.edu}
}
\date{}

\begin{document}
\maketitle

\begin{abstract}
Kidney exchange creates gains by pooling patient–donor pairs across hospitals and countries, but cooperation may unravel when coalitions can profitably withdraw. We introduce the \emph{supplemented core}, in which the platform uses voluntarily registered altruistic donors to restore stability. Because an altruistic donor adds a kidney without adding another patient, it provides an in-kind instrument for relaxing participation constraints when monetary transfers are unavailable. In worst-case compatibility graphs, the required number of donors can grow linearly with market size, even when cycle length is unrestricted. Under a standard heterogeneous random-compatibility model, however, a logarithmic number suffices for any fixed cycle-length bound, in expectation and with high probability.

Calibrated simulations find donor-free, maximum-cardinality weak-core exchanges in virtually all markets. The main challenge is selection, not existence: a representative lexicographic heuristic reflecting priorities used by kidney-exchange programs selects an unstable exchange in up to 35\% of markets even when a stable alternative exists. A small reserve of altruistic donors eliminates this implementation gap without sacrificing the heuristic’s operational objectives. Thus, altruistic donors do more than increase transplants; they sustain cooperation.
\end{abstract}

\smallskip
\noindent\textbf{Keywords:} kidney exchange; core stability; altruistic donors; cooperative games; matching markets


\section{Introduction}\label{sec:introduction}
Kidney exchange creates substantial gains by bringing together patients who have willing but incompatible living donors. These gains depend on market thickness. A larger common pool creates more opportunities for compatible cycles and chains, but it also requires hospitals to place their donor--patient pairs in the common market. In practice, kidney exchange remains fragmented. Hospitals may withhold pairs because participation involves administrative, testing, and coordination costs \citep{agarwal2019market}. Even if these costs are reimbursed, a hospital may prefer to conduct some exchanges internally rather than accept the allocation chosen by the centralized platform. The same issue arises in international kidney exchange, where countries decide whether to participate in a common program \citep{biro2019building}. The problem is therefore not only how to find more transplants, but also how to make participation attractive to the organizations that control the donor--patient pairs.

The core provides a natural benchmark for this participation problem. An exchange is in the weak core if no coalition of organizations can use the donor--patient pairs and altruistic donors under its control to make every member strictly better off. Core stability is stronger than individual rationality: it protects not only against the withdrawal of one organization, but also against coordinated withdrawal by several organizations. This is important when hospitals or countries can rematch their pairs outside the centralized program. The difficulty is that the weak core may be empty. Because kidney exchange is discrete, there may be no feasible exchange that satisfies all coalitional participation constraints.

This paper proposes a new way for the platform to restore stability: supplement selected exchanges with voluntarily registered altruistic donors. Altruistic donors are usually valued for initiating chains and increasing the number of transplants; we identify a distinct role. Because an altruistic donor contributes a kidney without adding another patient who must receive one, the donor gives the platform additional transplant capacity that can be used to satisfy an organization’s participation constraint. Since monetary payments are prohibited and cannot substitute for a kidney, this capacity acts as an in-kind subsidy.

The platform therefore faces a resource-allocation problem: it must deploy a scarce reserve of altruistic donors to satisfy the participation constraints that threaten cooperation, while using as few donors as possible. One additional transplant can have a disproportionate effect because the same organization may be pivotal in several potential blocking coalitions; improving that organization’s allocation can eliminate several deviations at once. Our analysis characterizes the size of the donor reserve needed to guarantee a core-stable exchange and identifies when this requirement remains small relative to the size of the market. 



We formalize this question through the \emph{supplemented core}. We say that an exchange belongs to the $d$-supplemented core if the centralized platform can use at most $d$ additional altruistic donors and the resulting exchange is not blocked by any coalition of participating organizations. The usual core is the special case $d=0$. The smallest value of $d$ measures how many additional donors are needed to make stable participation feasible.

We model kidney exchange as a directed graph. Each vertex represents a donor–patient pair, an unpaired patient, or a non-directed altruistic donor, and a directed edge indicates that one donor can give a kidney to another vertex's patient. Each organization controls a distinct set of donor--patient pairs and may also control its own altruistic donors. Feasible exchanges consist of cycles and chains, and each organization values an exchange by the number of its patients who receive a transplant.

The additional donors used in the supplemented core are coordinated by the centralized platform, but they are not owned by it. They should be interpreted as non-directed donors who voluntarily register with the centralized program and agree to have their recipients selected through its matching process. This arrangement is already used in practice. For example, non-directed donors may enter the OPTN Kidney Paired Donation Pilot Program independently, while the national program performs the matching and participating transplant centers remain responsible for evaluation and transplantation.\footnote{See the OPTN Kidney Paired Donation Pilot Program, \url{https://www.hrsa.gov/optn/patients/organ-donation/living-donation/kidney-paired-donation-pilot-program-kpdpp}.} Organizations may also have altruistic donors registered with them; these donors remain part of their endowments and are available when coalitions consider deviating. The additional donors are unavailable to a deviating coalition because they registered with the common program rather than with any member of that coalition. As in other static models of kidney exchange, the analysis is conditional on the set of donors who remain active at the time of the match. Every donor retains the right to withdraw.


We present three main theoretical results and use simulations to study their practical implications. The analysis separates two sources of instability. The first is intrinsic: a stable integral exchange may not exist. The second is algorithmic: a stable exchange exists, but the platform's allocation rule selects a different exchange that can be blocked. Additional altruistic donors address both problems. They restore existence when the core is empty and provide the flexibility needed to stabilize the outcome selected by an operational allocation rule.

Our first result shows that the problem can be severe. Without restrictions on the compatibility graph, the number of additional donors required for stability can grow linearly with the size of the market. This lower bound holds for pairwise exchange, for three-way exchange, and even for arbitrary cycle length bounds. Core nonexistence is therefore not merely a consequence of short cycle-length limits. In worst-case instances, the trivial policy of providing a separate donor for every unmatched patient is optimal up to a constant factor.

Our second result shows that {in pairwise exchange}, the number of additional donors required is at most the maximum number of \emph{independent odd cycles} in the mutual compatibility graph. A collection of vertex-disjoint odd cycles is independent if no mutual compatibility edge connects different cycles. These cycles capture the obstruction to converting a fractional stable allocation into an integral exchange.

This characterization yields a logarithmic guarantee in a finite-group random compatibility model. Pairs are partitioned by observable characteristics, and compatibility edges arise independently with group-specific probabilities. If the number of groups is fixed and all positive compatibility probabilities are bounded away from zero, then the maximum number of independent odd cycles—and hence the number of additional donors required—is $\mathcal{O}(\log |V|)$ in expectation and with high probability.

If the mutual compatibility graph is bipartite, there are no odd-cycle obstructions, and the weak core is nonempty without additional donors. We show that the bipartite result applies to a broad class of one-sided exchange economies with \emph{binary assignment valuations}, extending \citet{echenique2024stable}.


Our third result extends the logarithmic guarantee beyond pairwise exchange. For any fixed cycle-length bound $\Delta\geq 3$, under the same heterogeneous random compatibility model, the expected number of additional donors is {still} $\mathcal{O}(\log |V|)$, and the same bound holds with probability approaching one. The implicit constant depends on $\Delta$, the number of groups, and the smallest positive compatibility probability, and may be larger than in the pairwise case. The important point is that the donor requirement remains logarithmic rather than linear in market size. Thus, the linear worst-case lower bound is not representative of this broad class of random compatibility environments, even when exchanges involve longer cycles.

The theoretical bounds count the suitably compatible donors used in the final exchange; they do not assume that the platform can choose donor characteristics. In practice, the required compatibilities must be found among donors who voluntarily register, so the number screened may exceed the number ultimately used. If each newly registered donor has an independent probability bounded away from zero of satisfying a required compatibility, then the expected number screened is at most a constant multiple of the number used. The logarithmic bounds therefore remain logarithmic when expressed in terms of the size of the donor pool. Our simulations make no favorable selection assumption: altruistic donors are generated from empirically calibrated distributions and introduced in random order.

We complement the theoretical analysis with simulations based on a state-of-the-art kidney-exchange instance generator \citep{delorme2022improved,pettersson2021kidney}. We first ask whether there exists a maximum-cardinality exchange that is stable {against deviations by arbitrary coalitions of  organizations.}  Under this benchmark, we find a maximum-cardinality exchange that is weak core stable against these coalitions without additional donors in 1439 of the 1440 simulated markets.  {In the remaining market, at most two randomly generated donors suffice under all tested donor orders.} The required supplementation remains moderate under the stronger notions of the strong core and the transferable-utility (TU) core.\footnote{For these stronger core notions, we restrict blocking coalitions to at most nine organizations. Larger deviations are harder to coordinate and increasingly costly to test.}

Existence, however, does not imply stable implementation. Kidney-exchange programs typically select among feasible exchanges using a hierarchy of objectives, such as maximizing the number of transplants, favoring shorter cycles, preserving same-blood-type transplants, and matching hard-to-match patients \citep{biro2021modelling}. When we apply a representative lexicographic hierarchy without imposing core constraints, the selected exchange is blocked in as many as {35\%} of the simulated markets in some settings. Because a donor-free weak core exchange exists in almost all of these markets, these failures arise primarily from algorithmic selection rather than from core nonexistence. A market may possess a stable and efficient exchange and nevertheless be operated unstably.

We next retain the same lexicographic optima and introduce additional altruistic donors only when {the floor-preserving search does not find a stable donor-free optimum}. After each donor is added, {we recompute whether a core allocation that satisfies the lexicographic lower bounds from the optimal exchange exists}. The number of additional donors introduced by the procedure is at most \emph{three} in every parameter setting, {and at most $1.25$ on average even for the worst tested donor orders}. Supplementation therefore closes the gap between stable existence and stable implementation without requiring the platform to replace or reorder its established medical and operational objectives. It acts as a stability layer around the existing allocation rule.

We implement these procedures through an integer-programming framework. Our theoretical guarantees rely on Scarf's lemma and a rounding argument and therefore do not directly provide a practical algorithm. The computational framework uses a cycle formulation and searches for blocking coalitions. For the core-aware procedure, it iteratively adds the associated stability constraints and introduces altruistic donors when the constrained problem becomes infeasible. {For the practice-based procedure, it uses the lexicographic rule to compute optimum objective values, and adds a donor after the current cut-augmented search becomes infeasible.} The framework solves large instances efficiently on standard hardware and can certify {core} stability against {every} coalition.

Taken together, the results show that stability is a property of both the market and its allocation rule. The compatibility graph determines whether a stable exchange exists; the algorithm determines whether one is selected. Supplementation repairs both failures: it restores existence when the core is empty and stabilizes the outcome when the operational rule selects outside the core.

This gives altruistic donors a role beyond creating transplants and initiating chains: they provide the flexibility needed to sustain participation by hospitals and countries. Approximately 300--400 altruistic donors participate annually in the United States, compared with about 6,000 donor--patient pairs in the UNOS database; the NKR data studied by \citet{agarwal2019market} contain 164 altruistic donors and 1,265 pairs. Our proposal requires neither centralized ownership nor compulsory transfer of these donors. It requires only that some non-directed donors voluntarily participate through the common platform. A small subset can then serve as a stability reserve when the selected exchange would otherwise be blocked. The organizational value of such a donor can exceed the single transplant directly enabled: by preventing withdrawal, the donor preserves the larger gains from pooling.


The paper is organized as follows. After discussing related work, we introduce the model, the supplemented core, and the methodological preliminaries used to establish our results. Section~\ref{sec:low} provides worst-case lower bounds on the number of altruistic donors required. Section~\ref{sec:pairwise} studies pairwise exchange, while Section~\ref{sec:cycle} analyzes cyclic exchange.  Section~\ref{sec:simulation} presents the computational analysis. Section~\ref{sec:bipartite} presents the application to one-sided exchange economies. Section~\ref{sec:conclude} concludes. Proofs omitted from the main text are collected in the Appendix.


\subsection*{Related Literature}

The success of kidney exchange has inspired a rich literature on market design, beginning with \citet{roth2004kidney}. This literature studies different exchange structures, including pairwise and cyclic exchanges, and different strategic actors, ranging from individual donor--patient pairs to hospitals and countries. Incentives at the individual level are relatively well understood \citep{roth2004kidney,sonmez2013market}. Participation becomes more difficult when organizations control several pairs and can conduct exchanges outside the centralized program \citep{ashlagi2012new,ashlagi2014free,agarwal2019market}.

The main difference between our paper and the literature on kidney exchange among organizations is our positive approach to core stability. Existing work either emphasizes the possible nonexistence of the core or imposes weaker participation requirements. For example, \citet{ashlagi2014free} focus on individual rationality, while \citet{ashlagi2015mix} study incentive problems and multiplicative efficiency guarantees without imposing coalitional stability. Our model builds most directly on \citet{csaji2024ntu} and \citet{biro2019generalized}, which establish negative results for the core. To the best of our knowledge, no previous paper quantifies how much additional donor supply is sufficient to restore the core in kidney exchange.

Our approach also relates to credit systems for international kidney exchange \citep{klimentova2021fairness,biro2019generalized}. These systems assign transplant targets to participating countries and carry deviations from the targets forward as credits. Credits address fairness across matching rounds. We instead study stability within a matching round: conditional on the currently available pairs and donors, no coalition should gain by leaving the program. The two approaches are complementary. Intertemporal credits compensate countries in future rounds, while altruistic donors provide additional flexibility in the current round.

There is also a large literature on dynamic kidney exchange, including \citet{anderson2017efficient,dickerson2014balancing,akbarpour2020thickness}. This literature mainly studies trade-offs among efficiency, waiting time, and market thickness. Our model is static because we focus on whether the allocation in a given matching round satisfies coalitional participation constraints.

From a computational perspective, the literature has developed several methods for finding good kidney-exchange allocations under operational constraints \citep{mincu2021ip,ashlagi2024designing,druzsin2024performance}. These methods are mainly evaluated through simulations and generally do not provide theoretical guarantees for core stability. We obtain our theoretical results using the Relax-and-Round approach of \citet{nguyen2018near,nguyen2021stability}. Scarf's lemma provides a fractional stable benchmark, which we round to an integral exchange while bounding the additional donor supply required. Unlike work in which agents have small endowments \citep{jantschgi2025competitive}, hospitals in our model may control many donor--patient pairs. The rounding argument must therefore use the structure of the compatibility graph and the restrictions on feasible exchanges. We complement these existence results with an integer-programming framework for computing and verifying stable exchanges in realistic instances.

\section{Model and Preliminaries}\label{sec:model}

In this section we describe our mathematical model for kidney exchange, motivated by the works of \citet{biro2019generalized} and \citet{csaji2024ntu}.  
Unlike \citet{roth2004kidney}, which models each agent as a donor-patient pair, we consider agents as organizations, such as hospitals in the U.S. exchange system or organizations in the international exchange system. 
We consider an exchange involving a set of \( n \) organizations, denoted by \( \N =\{1,\dots ,n\}\). Each organization \( i \in \N \) is endowed with a set \( V^i \) consisting of donor-recipient pairs, unpaired patients, and altruistic donors. The sets \( V^i \) are disjoint across organizations. Within each \( V^i \), let \( U^i \subset V^i \) denote the set of 
non-altruistic  donor-patient pairs and unpaired patients.  The remaining elements, \( V^i \setminus U^i \), are altruistic donors  willing to donate without requiring a kidney in return.

\paragraph{Compatibility graphs.}
Compatibility is modeled via a directed graph, denoted by  \( \mathcal{G} = (V, E) \), where  $V=V^1\cup ..\cup V^n$,  and the set of edges, $E$, captures the compatibility of exchange. Specifically, 
\begin{itemize}
    \item For two non-altruistic donor—patient pairs $u, v$, a directed edge \( (u, v) \in E \) indicates that the donor in pair \( u \) is compatible with the patient in pair \( v \), meaning a transplant from \( u \) to \( v \) is feasible.  
    \item To simplify notation and the description of feasible exchanges, we can model altruistic donors as special donor–patient pairs by introducing a dummy patient who is compatible with any donor. This allows us to define directed edges following the same rules as before.
    \item Similarly, we can model unpaired patients by adding  a dummy donor who is only compatible with the dummy patient associated with the altruistic donors.
    
\end{itemize}

We examine the following model of \textit{ random compatibility graphs}~\citep{delorme2022improved}. 
Donor-recipient pairs are partitioned into a finite set of groups \( \Phi \), based on characteristics such as blood type and the sensitization levels of both the donor and the patient. 

For any two distinct donor-recipient pairs $u$ and $v$ belonging to
groups $a$ and $b$, respectively, the directed compatibility edge
$(u,v)$ is added independently with probability $p_{ab}$. Self-arcs
$(u,u)$ are excluded, but we impose no restriction on the diagonal
probabilities $p_{aa}$, which describe compatibility between distinct
pairs belonging to the same group. We write
\[
p:=\min\{p_{ab}:p_{ab}>0\}.
\]
The number of groups $g:=|\Phi|$ and the minimum positive probability
$p$ are treated as constants.

\paragraph{Exchanges.}
We consider two models of kidney exchange: \textit{pairwise exchange} and \textit{cyclic exchange}. In pairwise exchange, only mutually compatible donor-patient pairs can trade kidneys. Cyclic exchange allows longer cycles, where each donor gives to the next patient and receives a kidney from the previous donor. We assume a maximum cycle length of \( \Delta \), with \( \Delta = 2 \) corresponding to pairwise exchange. In practice, \( \Delta = 3 \) is typical, while values above 5 are rare due to the logistical and medical challenges of coordinating multiple simultaneous surgeries.\footnote{Even in the case of altruistic donors, cycles become chains, and operational difficulties with chains are less severe. Longer chains are also rare in practice (\citet{agarwal2019market}). Moreover, chains of size 10–20 typically result from the accumulation of exchanges over an extended period. For example, the National Kidney Registry reported a 25-transplant chain that was built up over several months. Our model aims to capture the clearing mechanism over a relatively shorter period, during which longer chains can be broken into smaller segments. The donor at the end of a chain in one period can then be viewed as an altruistic donor initiating a new chain in the next period. }

Let $\mathcal{C}_{\Delta}$ denote the cycles of length at most $\Delta$. Given a subset of vertices $W\subseteq V$, denote the induced compatibility graph on $W$ as  $\mathcal{G}[W]$. 
We call $\Ex$ an exchange among $W$ if $\Ex$ is a union of disjoint  cycles of length at most $\Delta$  in the induced compatibility graph  $\mathcal{G}[W]$.  We say that $y^*\in [0,1]^{\mathcal{C}_{\Delta}}$ is a \emph{fractional exchange }if it holds that for any vertex $v\in V$, 
$\sum\limits_{c\in \mathcal{C}_{\Delta}\mid v\in c}y^*_c\le 1.$

Given an exchange  $\Ex=(V_{\Ex}, E_{\Ex})$, the utility of an organization $i$ is given by the number of its patients receiving a kidney, that is  
$$
u_i(\Ex):=|V_{\Ex}\cap U^i|. 
$$ 

Overall, we refer to the economy as a \emph{partition exchange economy}, represented by the tuple $(\G, \Delta, \{V^i, U^i\}_{i=1}^n)$.

\begin{definition}
    Given a partition exchange economy, an exchange $\Ex$ is \emph{individually rational}, if for any  exchange $\Ex_i \subseteq \G [V^i]$, $u_i(\Ex)\ge u_i(\Ex_i)$ . That is, no player can improve its utility with an exchange of its own.
\end{definition}

\begin{definition}
Given a partition exchange economy, an exchange $\Ex$ among $V$,  and a coalition $P\subseteq N=\{1,..,n\}$, we say that $P$ is a \emph{(strong) blocking coalition to $\Ex$}, if there exists an exchange $\Ex'$ in the graph $\mathcal{G}[\cup_{i\in P}V^i]$ such that $u_i(\Ex')>u_i(\Ex)$ for all $i\in P$. We say that $\Ex$ is \textit{in the (weak) core}, if there is no blocking coalition to $\Ex$. 

A coalition $P\subseteq N$ \emph{weakly blocks} an exchange $\Ex$ if there exists an exchange $\Ex'$ feasible for $P$ such that
$
u_i(\Ex')\ge u_i(\Ex)\text{ for every }i\in P,
$
and
$
u_j(\Ex')>u_j(\Ex)\text{ for at least one }j\in P.
$ An exchange $\Ex$ is in the \emph{strong core} if it is not weakly blocked by any coalition $P\subseteq N$.
\end{definition}

Note that in the definition above, some papers in the literature refer to such a block as a strong block and the corresponding concept as the weak core. For convenience, in this paper we often refer to the weak core as simply the core. 

As we show in this paper, the core can be empty in both pairwise and cyclic exchanges. To obtain a positive result, we either impose a restriction on the compatibility graph—discussed in Section~\ref{sec:bipartite}—or relax the notion of the core by introducing what we call the supplemented core. To this end, we introduce a special agent, indexed by $0$, representing the centralized market designer, with  \( V^0 \) denoting the designer’s endowment of \emph{altruistic} donors.  Let \( \GG \) be the extended compatibility graph that incorporates the additional donors in \( V^0 \). The designer \emph{is not a player of the game and is not allowed to participate in any coalition}. Their sole role is to stabilize the exchange by providing additional resources when necessary.

Next we introduce our main concept.
    
\begin{definition} Given additional altruistic donors \( V^0 \), an exchange \( \Ex^* \) in the extended compatibility graph \( \GG \) is a \( V^0 \)-\emph{supplemented core} if it is not blocked by any coalition \( P \subseteq N =\{1,..,n\} \). {The \( d \)-\emph{supplemented core} refers to the pairs $(V^0,\Ex^*)$, such that $\Ex^*$ is a  \( V^0 \)-supplemented core exchange and $|V^0|\le d$.}
\end{definition} 


    

\begin{definition}
Given a set of additional altruistic donors \( V^0 \), an exchange \( \Ex^* \) in the extended compatibility graph \( \GG \) is \emph{Pareto optimal} if there does not exist another exchange \( \Ex' \) in \( \GG \) such that all organizations are weakly better off and at least one organization is strictly better off. 
\end{definition}

We show in a simple example how an additional altruistic donor can guarantee the existence of a core exchange.

\begin{example}
    Consider the partition economy illustrated in Figure~\ref{fig:pairwise-nocore}, featuring three organizations represented by red triangle,  green square, and  blue circle, respectively.  This instance, from \citet{csaji2024ntu}, demonstrates an empty core. However, adding a single altruistic donor (the black vertex) makes the bold matching $\Ex$ a core outcome. To see this, observe that any organization alone can cover all but 3 of its vertices, any two organizations can cover all but 2 of their vertices and the three organizations together can cover all but 5 of their vertices. Since in the bold $V^0$-supplemented exchange $\Ex$, one organization has 2 vertices unmatched, and the others only one, no single organization, nor the 3 organizations together block $\Ex$. Furthermore, no 2 organizations block either, as then they could leave at most $(2-1)+(1-1)=1$ of their vertices unmatched, which is impossible.
\end{example}

\begin{remark}
As this example illustrates, and as our main results formalize, stability can be restored by adding altruistic donors who are compatible with patients left unmatched by the rounding procedure. The deterministic bounds count the compatible donors ultimately used and therefore abstract from the process of finding them. In practice, the platform cannot choose donors' compatibility profiles. Suppose, however, that each newly registered donor, independently of previous arrivals, has a probability bounded away from zero of satisfying any outstanding compatibility requirement. In this case, the expected number of donors that must be screened for each successful match is bounded by a constant. Consequently, the expected number of donors that must be screened is at most a constant multiple of the number ultimately used, so the logarithmic bounds established below remain logarithmic. Our simulations impose this more realistic restriction: altruistic donors are generated randomly and added in random order. {We use three random orders, and distinguish donors introduced or screened from donors actually used in the implemented exchange.}
\end{remark}

\Xomit{
\begin{figure}
    \centering
    \includegraphics[width=0.4\linewidth]{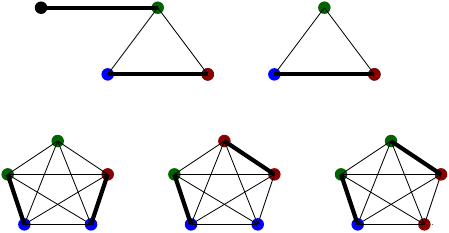}
    \caption{A partition exchange economy with pairwise exchanges without a core, but with a $1$-supplemented core. The red, green and blue vertices denote the organizations, while the black vertex in the upper left corner is $V^0$. The bold matching shows a $V^0$-supplemented core.}
    \label{fig:pairwise-nocore}
\end{figure}
}

\begin{figure}[h!]
\centering
\begin{tikzpicture}[scale=0.75, transform shape,
    base/.style={minimum size=7pt, inner sep=0pt},
    black_star/.style={base, star, star points=5, fill=red, scale=2.5},
    green_square/.style={base, rectangle, fill=green!50!black, scale=1.5},
    red_triangle/.style={base, regular polygon, regular polygon sides=3, fill=black, scale=2.5},
    blue_circle/.style={base, circle, fill=blue, scale=1.7},
    thin_edge/.style={draw, thin},
    thick_edge/.style={draw, line width=2pt}
]

    \begin{scope}[shift={(0,4)}]
        \node[black_star] (n1) at (0,1.5) {};
        \node[green_square] (n2) at (2,1.5) {};
        \node[blue_circle] (n3) at (1,0) {};
        \node[red_triangle] (n4) at (3,0) {};
        
        \draw[thick_edge] (n1) -- (n2);
        \draw[thick_edge] (n3) -- (n4);
        \draw[thin_edge] (n2) -- (n3);
        \draw[thin_edge] (n2) -- (n4);
    \end{scope}

    \begin{scope}[shift={(5,4)}]
        \node[green_square] (m1) at (1,1.5) {};
        \node[blue_circle] (m2) at (0,0) {};
        \node[red_triangle] (m3) at (2,0) {};
        
        \draw[thick_edge] (m2) -- (m3);
        \draw[thin_edge] (m1) -- (m2);
        \draw[thin_edge] (m1) -- (m3);
    \end{scope}

    \begin{scope}[shift={(0,0)}]
        \foreach \a [count=\i] in {90,162,234,306,18} \coordinate (p\i) at (\a:1.5);
        \foreach \i in {1,...,5} \foreach \j in {\i,...,5} \draw[thin_edge] (p\i) -- (p\j);
        
        \node[green_square] at (p1) {}; \node[green_square] at (p2) {};
        \node[blue_circle] at (p3) {}; \node[blue_circle] at (p4) {};
        \node[red_triangle] at (p5) {};
        
        \draw[thick_edge] (p2) -- (p3); \draw[thick_edge] (p4) -- (p5);
    \end{scope}

    \begin{scope}[shift={(4.5,0)}]
        \foreach \a [count=\i] in {90,162,234,306,18} \coordinate (q\i) at (\a:1.5);
        \foreach \i in {1,...,5} \foreach \j in {\i,...,5} \draw[thin_edge] (q\i) -- (q\j);
        
        \node[red_triangle] at (q1) {}; \node[green_square] at (q2) {};
        \node[blue_circle] at (q3) {}; \node[blue_circle] at (q4) {};
        \node[red_triangle] at (q5) {};
        
        \draw[thick_edge] (q2) -- (q3); \draw[thick_edge] (q1) -- (q5);
    \end{scope}

    \begin{scope}[shift={(9,0)}]
        \foreach \a [count=\i] in {90,162,234,306,18} \coordinate (r\i) at (\a:1.5);
        \foreach \i in {1,...,5} \foreach \j in {\i,...,5} \draw[thin_edge] (r\i) -- (r\j);
        
        \node[green_square] at (r1) {}; \node[green_square] at (r2) {};
        \node[blue_circle] at (r3) {}; \node[red_triangle] at (r4) {};
        \node[red_triangle] at (r5) {};
        
        \draw[thick_edge] (r2) -- (r3); \draw[thick_edge] (r1) -- (r5);
    \end{scope}
\end{tikzpicture}
\caption{A partition exchange economy with pairwise exchanges that has no core but admits a $1$-supplemented core. The three organizations are represented by vertices with distinct shapes (circle, triangle, square). The star vertex in the upper-left corner represents $V^0$. The bold matching belongs to the supplemented core.}
 \label{fig:pairwise-nocore}
\end{figure}

It is well known that in exchange economies with satiated utilities, the strong core is an unreasonably demanding solution concept. 

\begin{remark}\label{rem:score}


Consider the following example. Organization $A$ has $k$ altruistic donors and no patients. There are $k+1$ other organizations, $B_1,\ldots,B_{k+1}$, each owning $k$ donor--patient pairs. Every altruistic donor of $A$ is compatible with every patient of every $B_i$, and the donor in each pair owned by a $B_i$ is compatible with the dummy patient associated with every altruistic donor. No exchange is possible among the organizations $B_1,\ldots,B_{k+1}$ without the altruistic donors. In any exchange using only the $k$ donors initially owned by $A$, at least one organization $B_i$ has fewer than $k$ patients matched. The coalition $\{A,B_i\}$ then weakly blocks: $A$ remains indifferent because it has no patients, while $B_i$ can use all $k$ donors of $A$ to match all of its patients. Hence, in any strong-core allocation, every organization $B_i$ must have all $k$ of its patients matched. This requires $k(k+1)$ altruistic donors in total. Since $A$ initially owns $k$ such donors, at least $k^2$ additional altruistic donors are necessary. Conversely, $k^2$ additional universally compatible altruistic donors are sufficient. Thus, the number of additional donors required for the strong core is $\Theta(k^2)$, which is of the same order as the total number of vertices in the original instance, namely $k^2+2k$.  By contrast, no additional donor is needed for the weak core. Since organization $A$ has no patients, it can never be strictly better off, and no coalition consisting only of the organizations $B_1,\ldots,B_{k+1}$ can perform any exchange. Hence any maximum-cardinality allocation of the original $k$ altruistic donors is Pareto optimal and belongs to the weak core. {Note also that this example can arise with only two groups $a$ (group of the pairs in the $B_i$-s) and $b$ (group of the altruist-dummy pairs) with $p_{ab}=p_{ba}=1$, $p_{aa}=0$.}

\end{remark}

\medskip
 The requirement in Remark~\ref{rem:score} is clearly unrealistic. Therefore, we focus on the weak core and Pareto optimality, rather than the strong core. Remark~\ref{rem:score} also shows that to have a strong core solution, we may need to add $|V^i|$ many additional donors for some of the organizations $i$ -- note that adding $|V^i|$ many additional donors for $i$ ensures that $i$ cannot strictly improve anymore.

\subsection{Supplementary Lemmas}\label{sec:scarf}

Our results rely on Scarf’s lemma together with a rounding argument from \citet{nguyen2018near}. We prove below three technical lemmas that underpin the analysis. The first is Scarf’s lemma, the second shows how to round a fractional solution to an integral one, and the third is a basic linear-algebraic fact from linear programming that we use to analyze the quality of the rounding procedure.

For a matrix $Q$, let $Q_{ij}$ denote the $j$-th element of the $i$-th row of $Q$, $Q_i$ denote the $i$-th row of $Q$ and $Q_{\cdot j}$ denote the $j$-th column of $Q$.

\begin{lemma}[\citet{scarf1967core}] Let $Q$ be an $n\times m$ nonnegative matrix, such that every column of $Q$ has a nonzero element and let $q\in \mathbb{R}^n_+$. Suppose that every row $i$ has a strict ordering $\succ_i$ on those columns $j$ for which $Q_{ij}>0$. Then there is an extreme point of $\{x\in \mathbb{R}^m\mid  Qx\le q, \; x\ge 0\}$, that dominates every column, where we say that $x$ dominates column $j$ if there is a row $i$ with 
$Q_{ij}>0$, $Q_i x=q_i$, and, for every column $k$ such that
$Q_{ik}x_k>0$, either $k=j$ or $k\succ_i j$. 
\end{lemma}

Let $\gamma_c^i=|U^i\cap c|$ for a cycle $c$ and organization $i$.
Recall that  $\mathcal{C}_{\Delta}$ denotes the set of cycles of length at most $\Delta$. The following key lemma -- building on Scarf's Lemma -- provides an existential guarantee for a special fractional solution: one that we can round to find a $V^0$-supplemented core.

\begin{lemma}[Rounding Lemma]\label{thm:ApplyScarf}
   Given a partition exchange economy with cycles of size at most 
   $\Delta$, there exists a fractional exchange $y^*$ in $\G$ such that: for any set $V^0$ of altruistic donors and any (deterministic) exchange $\Ex^*$ in $\GG$, the exchange $\Ex^*$ belongs to the $V^0$-supplemented core if
   $$
   u_i(\Ex^*)\ge \big \lfloor \sum\limits_{c\in \mathcal{C}_{\Delta}}\gamma_c^i y^*_c  \big\rfloor  \text{ for all } i\in \N.
   $$
\end{lemma}

\begin{proof}{Proof.}
Let us create an instance of Scarf's Lemma as follows. For the matrix $Q$, we create a row for each vertex $v\in V$. Furthermore, for each coalition $\varnothing\ne P\subseteq \N$ and each nonzero utility vector $(u_1, \dots, u_{|P|})$ which is attainable by a feasible exchange $\varnothing\ne \Ex \subseteq E[P]$ (that is $u_i(\Ex) =u_i$ for $i\in P$), we create a column $(P,\Ex)$ by choosing an arbitrary such exchange $\Ex$. The entry in the intersection of a row $v$ and a column $(P,\Ex)$ is 1, if $v$ is covered by $\Ex$ (also implying that $v\in V^i$ for some $i\in P$ by $\Ex \subseteq E[P]$) and $0$ otherwise. The bounding vector $q$ is chosen to be $1$ for all rows. Both are nonnegative, hence they satisfy the conditions of Scarf's Lemma.

Then, we create the strict ranking required for the rows over their nonzero elements. For a vertex $v\in V^i$, this is created in a way such that $(P,\Ex)\succ_i(P',\Ex')$, whenever it holds that $u_i(\Ex)>u_i(\Ex')$. If $u_i(\Ex)=u_i(\Ex')$, then we break the ties in a way such that $(P,\Ex)\succ_i(P',\Ex')$, whenever $|P|<|P'|$. Any remaining ties are broken according to an arbitrary fixed ordering of the columns.

By Scarf's Lemma, there exists an extreme point  $x^*$ of this system that dominates every column.

Define $y^*_c=\sum\limits_{P,\Ex\mid c\in \Ex}x^*_{P,\Ex}$ for all cycles $c\in \mathcal{C}_{\Delta}$. Then, 
$$
\sum\limits_{c\mid v\in c}y^*_c=\sum\limits_{c\mid v\in c}\sum\limits_{P,\Ex\mid c\in \Ex}x^*_{P,\Ex}=\sum\limits_{P,\Ex\mid \Ex \text{ covers } v}x^*_{P,\Ex} = Q_vx^*\le 1.
$$ 
Thus, $y^*$ is indeed a fractional exchange.

Let $\Ex^*$ be an exchange using some $d\ge 0$ altruistic donors, such that $u_i(\Ex^*)\ge \lfloor \sum\limits_{c\in \mathcal{C}_{\Delta}}\gamma_c^iy^*_c\rfloor$ for all $i\in \N$. 

Suppose there is a blocking coalition $P$ with an exchange $\Ex_P$, where all $i\in P$ improve. By construction, there is an exchange $\Ex_P'$, such that $P,\Ex_P'$ has a column in $Q$ and $u_i(\Ex_P)=u_i(\Ex_P')$ for all $i\in P$. Hence, because $x^*$  dominates every column, it dominated  $P,\Ex_P'$ in some row $v$. Let $j$ be the organization for which $v\in V^j$. For this $j\in N$, we have that $\sum_{P',\Ex'}Q_vx^*_{P',\Ex'}=1$ and that for any $P',\Ex'$ with $j\in P'$ and $x^*_{P',\Ex'}>0$ it holds that $u_j(\Ex')\ge u_j(\Ex_P')$.

Hence, by the construction of $y^*$, we must have $\sum\limits_{c\in \mathcal{C}_{\Delta}}\gamma_c^jy_c^*=\sum\limits_{P',\Ex'}\sum\limits_{c\in \Ex'}\gamma_c^jx^*_{P',\Ex'}=\sum\limits_{P',\Ex'}u_j(\Ex')x^*_{P',\Ex'}\ge \sum\limits_{P',\Ex'\mid \Ex' \text{ covers }v} u_j(\Ex')x^*_{P',\Ex'}\ge Q_vx^*\cdot \min_{P',\Ex'} \{ u_j(\Ex')\mid Q_{v,(P',\Ex')}x^*_{P',\Ex'}>0\}\ge  u_j(\Ex_P')$. 

Therefore, we get that $u_j(\Ex^*)\ge \lfloor \sum\limits_{c\in \mathcal{C}_{\Delta}}\gamma_c^jy^*_c\rfloor \ge u_j(\Ex_P')=u_j(\Ex_P)$, contradicting that $P$ is a strong blocking coalition. 
\end{proof}

We will also use the following well-known lemma from operations research (\citet{schrijver1998theory}).
\begin{lemma}
\label{lemma:extremepoint}
Let $z^*$ be an everywhere strictly positive extreme point of $\{ Qz\le  q,  z\ge 0\}$. Then the number of variables equals the maximum number of linearly independent tight rows of $Q$.
\end{lemma}

\section{Lower Bounds}\label{sec:low}

In this section, we establish lower bounds on the number of additional donors required for pairwise exchange, three-way exchange, unbounded exchange, and $\Delta$-exchange for any fixed $\Delta$.


In the case of pairwise exchange, we can discard all compatibility arcs $(u,v)$ such that $(v,u)$ is not a compatibility arc, and only consider the undirected graph $\graph$ of the mutual compatibilities. Hence, in this case, we replace the set $E$ of the edges of $\graph$ by the possible pairwise exchanges $\mathcal{C}_2$. Therefore, the feasible exchanges are now assumed to be the matchings of the graph $\graph$.





\begin{figure}
    \centering
    \includegraphics[width=0.3\linewidth]{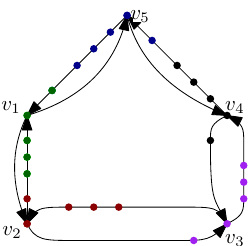}
    \caption{An instance, where the core is empty for any $\Delta\ge 6$, including $\Delta =\infty$. Colors denote the organizations.}
    \label{fig:empty-core}
\end{figure}

\begin{example}\label{ex:no3core}
   We give a simple example, where the weak core is empty for any bound $\Delta\ge 6$, in particular even for unbounded exchanges, illustrated in Figure~\ref{fig:empty-core}. (For pairwise exchanges, a no-instance was given by \citet{csaji2024ntu}). We have five organizations,
   $1,2,3,4,5$, each with five vertices in $V^i$. The compatibility graph is a five-cycle of mutual compatibilities with subdivided arcs, where each organization controls exactly one of the vertices of the five-cycle. The one controlled by $i$, we call $v_i$. Then, $v_2v_1$ is subdivided with four vertices, three belonging to organization $1$, and one to organization $2$. Arc $v_2v_3$ is subdivided with one vertex, belonging to organization $3$, and arc $v_3v_2$ is subdivided with three vertices, each belonging to organization $2$.
Arc $v_3v_4$ is subdivided with three vertices belonging to organization $3$, and arc $v_4v_3$ with one vertex belonging to organization $4$. Arc $v_4v_5$ is subdivided with four vertices, three belonging to organization $4$, and one to organization $5$. Arc $v_5v_1$ is subdivided with four vertices, three belonging to organization $5$, and one to organization~$1$. 

{If one of the large cycles containing all of $\{ v_1,v_2,v_3,v_4,v_5\}$ is chosen, then it is the only exchange cycle and either the coalition $\{1,2 \}$ or $\{4,5 \}$ blocks the exchange. Otherwise, if only cycles between consecutive $v_i$-s are chosen, then there exists some $j$ with $v_j$ uncovered. Then, the coalition $\{ j-1,j\}$ blocks the exchange.} 
\end{example}

\medskip
Note that by using disjoint copies of the example, we already obtain that $\Omega (|V|)$ donors may be necessary. 
Next, we show that for any bound $\Delta$, $\Omega (|V|)$ donors may be required for the core, even with three organizations.


\begin{theorem}\label{thm:general-empty-core}
    Fix any exchange bound $\Delta$. Then, there exists a constant $C_\Delta$ and arbitrarily large partition exchange economies such that the $C_\Delta |V|$-supplemented core is empty in them, even if the number of players is $n=3$.
\end{theorem}
\begin{proof}{\textbf{Proof for Pairwise Exchanges.}}  We create an instance of a partition exchange economy with pairwise exchanges as follows. Let $K_l$ be the clique on $l$ vertices - i.e. the graph where any two vertices are connected. The graph $\graph$ of the partition exchange economy consists of $3a$ vertex-disjoint $K_3$ graphs (i.e. triangles) $K_3^1,K_3^2, \dots, K_3^{3a}$ and $9a$ vertex-disjoint $K_5$ graphs $K_5^1,K_5^2,\dots, K_5^{9a}$, where $a\in \mathbb{N}$ is a parameter. We have that $|V|=9a+45a=54a$, and the size of each organization is $18a$.

 There are three organizations $1,2,3$. Each triangle $K_3^j$ contains one vertex from each organization. Furthermore, $K_5^{(j-1)\cdot 3a +1},\dots, K_5^{(j-1)\cdot 3a+ 3a}$ contains 2 vertices from organizations $j$ and $j+1$ and 1 vertex from $j+2$ for $j\in [3]$, where addition is taken modulo $3$ (i.e. $2+1=3$, but $2+2=1$).

It is easy to see that for any pairwise exchange (matching), at least $12a$ vertices must remain unmatched, at least one for each disjoint clique in the graph. 

Furthermore, if two organizations join, then we claim that together they can create a matching where at most $6a$ of their vertices remain unmatched. Indeed, they can match all their vertices in the triangles as there is always an edge between them. Among $K_5^1,\dots, K_5^{9a}$, there are $6a$ cliques, where one of them has an odd number of vertices. In such cliques, they must leave one of their vertices unmatched, but this can be an arbitrary one of their vertices, as any two vertices within a $K_5^j$ are connected. In the cliques $K_5^j$, where both of them have an even number of vertices, all can be matched. 

Finally, a single organization can also create a matching, leaving only $6a$ of its vertices unmatched. This holds because among the $9a$ cliques $K_5^1,\dots, K_5^{9a}$, it has an even number of vertices in $6a$ of them, so all of them can be matched there. Hence, only $3a$ vertices in the $K_5^j$s and $3a$ in the triangles will be unmatched.

Let $d:= a $. Suppose, to the contrary, that there exists a $3(d-1)$-supplemented core solution $\Ex$. By our observations above, the two worst-off organizations  together have at least $\frac{2}{3}\cdot (12a-(3d-3))=8a - 2d +2=6a+2$ unmatched vertices. Say one of them has $c\ge 0$ and the other at least $8a-2d+2-c$. As each organization can create a matching with at most $6a$ unmatched vertices, we must have that $c \le 6a$ and $8a-2d+2 - c\le 6a$, so $2a-2d +2\le c$. Hence, $2\le c\le 6a=8a-2d+2-2$ and so both organizations must have at least $2$ vertices unmatched. Hence, both organizations can improve their situation in a matching among themselves, where only $6a<6a+2$ vertices are unmatched, contradicting the fact that the exchange is in the $3(d-1)$-supplemented core. 

As we have that $d=a$ and hence $d=\frac{|V|}{54}$, the statement follows.

For $\Delta\ge 3$, the proof is deferred to the Appendix. 
\qed{}
\end{proof}


\section{Upper Bounds}\label{sec:upper}
\subsection{Pairwise Exchange}\label{sec:pairwise}


We first consider pairwise exchange. Our main result on the existence of the supplemented core relies on the following concept of independent odd cycles.
\begin{definition}
Given an undirected graph $\graph =(V,E)$, a set of odd-length cycles $\mathcal{C}=\{ c_1,\dots, c_l\}$ is called \emph{independent}, if $c_i\cap c_j=\emptyset$ for $i\ne j$ and there exists no $u\in c_i, v\in c_j$, $j\ne i$ such that $(u,v) \in E$. Let $\oddcycles(\graph)$ denote the maximum number of independent odd cycles of $\graph$.
\end{definition}

\begin{theorem}\label{thm:arbitrary}
    In a partition exchange economy with pairwise exchanges, the 
    $ \oddcycles(\graph)$-supplemented core
    is always nonempty and  contains a Pareto-optimal solution.   As a special case, if $\graph$ is bipartite, then the core is nonempty and contains a Pareto-optimal solution.
\end{theorem}

\begin{proof}{Proof.}    
By Lemma~\ref{thm:ApplyScarf}, there exists a fractional exchange $y^*$, such that if $u_i(\Ex^*)\ge \lfloor \sum_{c\in \mathcal{C}_{2}}\gamma_c^iy^*_c\rfloor = \lfloor\sum_{e\in E}\gamma_e^iy^*_e\rfloor$ for all $i\in N$ for some exchange $\Ex^*$ using $d$ additional donors, then $\Ex^*$ is in the $d$-supplemented core. Let $y^*(E(U^i)):=\sum_{e\in E}\gamma_e^iy^*_e$.

 We create a polyhedron $\mathcal{P}= \{ Az\le a, Bz\ge b, z\ge 0\}$ as follows. The variables $z$ correspond to the edges of $\graph$. Furthermore, the matrix $A$ is just the incidence matrix of $\graph$ with $a\equiv 1$. In the matrix $B$, we have a row for each organization, such that $B_{i,e}=\gamma_e^i=|e\cap U^i|$. It is easy to see that the row $B_i$ of an organization $i$ is the sum of the rows $A_v$ for $v\in U^i$. Every entry of $B$ is from $\{ 0,1,2\}$. The vector $b$ is set such that $b_i = \lfloor y^*(E(U^i))\rfloor$.  




The polyhedron $\mathcal{P}$ is clearly nonempty, as $y^*\in \mathcal{P}$.

Consider the following procedure. Let $z^*$ be an extreme point of $\mathcal{P}$ that maximizes $\sum\limits_{e\in E}(\sum\limits_{i\in \N}\gamma_e^i )z_e^*$. If there is an integer coordinate $z_e^*$ in $z^*$, then we delete that column from $A$ and $B$ and set $a:= a-z^*_e\cdot A_{\cdot,e}$ and $b:= b-z_e^*\cdot B_{\cdot,e}$, where $Q_{\cdot,j}$ denotes columns $j$ for a matrix $Q$. Also, if a row becomes all zero, then we delete it. Furthermore, if $a_v$ becomes 0, then we can set all variables with a nonzero entry in the row $A_v$ to $0$ and delete them. Hence, we can assume that $a$ is always $1$ for every coordinate. If some row or column got deleted, then we update $z^*$ to an extreme point of the new polyhedron $\mathcal{P}$ that maximizes $\sum\limits_{e\in E}(\sum\limits_{i\in \N}\gamma_e^i )z_e^*$.

If at some point, all variables become integral, then it gives an exchange $\Ex$ in the core by Lemma~\ref{thm:ApplyScarf}. 

Otherwise, suppose that all variables of the extreme point $z^*$ are fractional.
Then, we can use Lemma~\ref{lemma:extremepoint} for the (updated) polyhedron $\mathcal{P}$ to show some properties that are satisfied by $z^*$.

Credit each variable $z_e^*$ of $z^*$ with 1 token. Then, we redistribute the tokens as follows. If $u\in V^i\setminus U^i$ for some $i\in \N$, or $A_uz^*=1$, then we give $\frac{1}{2}$ token to $A_u$. Else, $u\in U^i$, for some $i$, $A_uz^*<1$ and we give $\frac{1}{2}$ token to $B_i$. Similarly, if $v\in V^j\setminus U^j$ for some $j\in \N$, or $A_vz^*=1$, then we give $\frac{1}{2}$ token to $A_v$. Otherwise, $v\in U^j$ for some $j$, $A_vz^*<1$ and we give $\frac{1}{2}$ token to $B_j$. It is easy to see that each $z_e^*$ component distributes at most 1 token.


Furthermore, each tight row $A_v$ gets at least one token, since each coordinate is $\{ 0,1\}$ in $A_v$ and all coordinates are fractional, so there are at least two edges $e_1,e_2$ that give $\frac{1}{2}$ token to $A_v$. 

Take a tight row $B_i$, such that not all $A_v$ rows with $v\in U^i$ are tight. Otherwise, the row $B_i$ would be generated by the tight rows $A_v$ for $v\in U^i$. 

Then, take a vertex $v\in U^i$ with $A_vz^*<1$. By the tightness of $B_i$, we have that there are at least two vertices $v\in U^i$ with $A_vz^*<1$, since $b_i=B_iz^* = \sum_{v\in U^i}A_vz^*$ and $b_i$ is integer. 
Therefore, $B_i$ receives at least one token too.

By Lemma~\ref{lemma:extremepoint} we get that this is only possible, if each tight row that received at least one token received exactly one token, and if a row received less than one token, then it received 0.

From this, we conclude that there are at most two fractional edges of $z^*$ incident to any vertex $v$ with $A_vz^*=1$. We claim that this is also true for a vertex $v$ with $A_vz^*<1$. Indeed, if $v\in U^i$ for an organization $i$ such that $B_i$ is tight, then all fractional edges incident to $v$ give a $\frac{1}{2}$ token to $B_i$'s row, so there are at most two.  Thus, the fractional support is a vertex-disjoint collection of paths and cycles in \(G\).

{If \(v \in V_i \setminus U_i\) and \(A_v z^* < 1\), then no fractional edge is incident to \(v\), since otherwise the non-tight row \(A_v\) would receive a positive token. Hence, every endpoint of a fractional path is utility-bearing, although tight altruistic vertices may occur internally on paths or cycles.} Also, no paths can be odd-length, as otherwise it contradicts the fact that $z^*$ maximizes $\sum\limits_{e\in E}(\sum\limits_{i\in \N}\gamma_e^i )z_e^*$. Furthermore, there can be no even length cycle $C$ either, since otherwise $z^*$ is not extreme: $z^* = \frac{1}{2}(z_1+z_2)$, where  $z_1=z^*-\varepsilon \chi^{C_1}+\varepsilon \chi^{C_2}$ and $z_2=z^*+\varepsilon \chi^{C_1}-\varepsilon \chi^{C_2}$ for some sufficiently small  $\varepsilon>0$ and $C=C_1\cup C_2$ is chosen such that the edges of $C$ alternate between $C_1$ and $C_2$.

{Finally, consider the even-length paths in the fractional support. Each endpoint $v$ of such a path has $A_vz^*<1$. Hence, $v\in U^i$ for some organization $i$ with tight row $B_i$, and the endpoint contributes $\frac12$ token to $B_i$. Since every tight row receiving tokens receives exactly one token, every such organization controls exactly two path endpoints. Construct an auxiliary multigraph whose vertices are these organizations and whose edges are the even-length paths, joining the organizations containing their endpoints. Every vertex of this auxiliary multigraph has degree two, so each component is a cycle (possibly a loop). Orient each such cycle consistently. On every corresponding even-length path, perturb its edges alternately by $+\varepsilon$ and $-\varepsilon$, choosing the signs so that the perturbation is $-\varepsilon$ at the endpoint belonging to the tail organization and $+\varepsilon$ at the endpoint belonging to the head organization. At every internal vertex the two changes cancel, and at every organization the changes at its two path endpoints cancel as well. Thus all tight $A$- and $B$-constraints are preserved. For sufficiently small $\varepsilon>0$, all non-tight constraints and nonnegativity constraints remain feasible. Applying the opposite perturbation gives two distinct feasible vectors $z_1,z_2$ with $z^*=\frac12(z_1+z_2)$, contradicting that $z^*$ is an extreme point.}

Hence, we conclude that the fractional edges of $z^*$ give a vertex disjoint set of odd cycles $\mathcal{C}$. 
If $\graph$ is bipartite, then this is a contradiction, which shows that the core is nonempty in this case, and $z^*$ corresponds to a Pareto-optimal exchange in the core. 

Suppose that $\graph$ is not bipartite.

We can suppose that $\mathcal{C}$ is an independent set of odd cycles. Otherwise, if there exists an edge $(u,v)$ with $u\in c_{j_1},v\in c_{j_2}$, $j_1\ne j_2$, then all vertices in $c_{j_1}\cup c_{j_2}$ can be covered with a matching. Hence, we may choose such a matching, and iterate this procedure until no such case exists.

From this, it follows that we can round $z^*$ to a matching $\Ex$ that satisfies that at most $\oddcycles (\graph) $ vertices are unmatched from these cycles. Hence, with $\oddcycles (\graph)$ additional donors, we can guarantee that the final supplemented exchange $\Ex^*$ satisfies $u_i(\Ex^*)\ge \lfloor y^*(E(U^i))\rfloor$,
so by Lemma~\ref{thm:ApplyScarf} it gives a solution in the $\oddcycles (\graph )$-supplemented core. If $\Ex^*$ is not Pareto-optimal, then we can take a Pareto-improving exchange that maximizes total utility among such improvements. It still satisfies the required utility lower bounds, so is in the $\oddcycles (\graph)$-supplemented core, and it is also  Pareto-optimal.\qed{}
\end{proof}

The proof of the following remark can be found in the Appendix.
\begin{remark}\label{rem:localchange}
The $\oddcycles (\graph )$-supplemented core exchange in Theorem~\ref{thm:arbitrary} can be chosen in such a way that each organization $i$ requires at most $\lceil \frac{|U^i|}{3}\rceil$ additional donors. 
\end{remark}

\medskip

Next, we show that in a random compatibility graph, the expected maximum number of independent odd cycles grows logarithmically with the size of the graph. This shows that in a typical random compatibility graph, compared with the market size, only a small number of altruistic donors are needed to stabilize the exchange.

\begin{lemma}\label{thm:logbound}
   Consider a compatibility graph that is generated according to the random graph model. Assuming a constant number $g$ of groups, for any distribution of the vertices to the groups,  in a random kidney exchange graph $\graph$, the expected value satisfies $\mathbb{E}[\oddcycles ( \graph)]=\mathcal{O}(\log |V|)$ and as the size of graph increases,  the probability that  $\oddcycles ( \graph)=\mathcal{O}(\log |V|)$ approaches 1. 
\end{lemma}

From Lemma~\ref{thm:logbound} and Theorem~\ref{thm:arbitrary} we immediately get the following logarithmic bound on the altruistic donors needed by the system.
\begin{theorem}
    { Consider a partition exchange economy with pairwise exchanges, where the compatibility graph on the original vertices $V$, including the existing altruistic donors in $V\setminus U$, is generated according to the random graph model with a constant number of groups. Then, there is a constant $C=C(g,p)$ such that, if $d(\graph)$ denotes the minimum number of \emph{additional} altruistic donors required for a supplemented-core exchange, then
\[
\mathbb E[d(\graph)]
\le
C\log|V| \;\;\; \text{ and } \;\;\;
\Pr\left(
d(\graph)\le C\log|V|
\right)\xrightarrow{|V| \to \infty} 1.
\] }
     {The expectation and probability are taken over the realization of the random compatibility arcs. The supplemented-core exchange can be chosen Pareto-optimal.}
\end{theorem}



\subsection{Cyclic Exchange}\label{sec:cycle}


Next, we consider cyclic exchanges. We first study the practically most
relevant case in which cycles have length at most three. In
Appendix~\ref{app:fixed-cycle-length}, we show that the same logarithmic
bound extends to every fixed upper bound on cycle length.

The pairwise argument from Section~\ref{sec:pairwise} does not extend
directly. With pairwise exchanges, the support of a fractional extreme
point can be reduced to independent odd cycles. With three-way
exchanges, fractional extreme points may have substantially more
complicated supports. Our argument therefore rounds the aggregate
numbers of cycles of different patterns rather than the support of the
fractional exchange itself.

\begin{theorem}\label{thm:cycle-random}
Consider a partition exchange economy with $\Delta=3$, where the
compatibility graph on the original vertices $V$, including the
existing altruistic donors in $V\setminus U$, is generated according
to the random compatibility model. Let $m:=|V|$, and let $d(\graph)$
denote the minimum number of additional altruistic donors required for
a supplemented-core exchange.  There is a constant $C=C(g,p)$ such that
\[
    \mathbb E[d(\graph)]
    \le C\log(m+2) \qquad  \text {and} \qquad
    \Pr\left(
        d(\graph)\le C\log(m+2)
    \right)
    \ge 1-(m+2)^{-4}.
\]
The expectation and probability are taken over the realization of the
random compatibility arcs. Moreover, the supplemented-core exchange
can be chosen Pareto-optimal.
\end{theorem}

The proof rests on the following rounding result. Its uniformity is
important because the fractional exchange obtained from
Lemma~\ref{thm:ApplyScarf} depends on the realized compatibility graph.

\begin{lemma}[Logarithmic aggregate-pattern rounding]
\label{lem:logarithmic-three-rounding}
Let $m:=|V|$. There is a constant $C_0=C_0(g,p)$ such that, with
probability at least $1-(m+2)^{-4}$, the following statement holds
simultaneously for every fractional exchange $y$ using cycles of size
at most three: there is an integral exchange $\Ex$, using only the
original vertices in $V$, such that
\[
    \sum_{i\in\N}
    \left(
        \left\lfloor
            \sum_{c\in\mathcal C_3}\gamma_c^i y_c
        \right\rfloor
        -
        u_i(\Ex)
    \right)^+
    \le C_0\log(m+2).
\]
\end{lemma}

The formal proof is given in
Appendix~\ref{app:three-way-rounding}. We describe its structure here.
A cycle pattern records, for each position in a two- or three-cycle,
the random-graph group of the corresponding vertex and whether that
vertex represents a utility-bearing patient or an existing altruistic
donor. Because the number of groups and the cycle-length bound are
fixed, there are only constantly many patterns.

For each pattern, we retain only the total fractional weight assigned
to cycles of that pattern and round this total down to an integer.
Flooring the pattern weights removes only a constant number of
utility-bearing cycle positions. A transportation argument then
assigns the remaining integer numbers of positions to organizations.
The integrality of the transportation problem ensures that this
assignment creates only a constant aggregate loss relative to the
organizations' fractional utilities. Existing altruistic donors enter
the relevant capacity constraints, but an uncovered altruistic slot is
not charged as a utility loss.

It remains to realize the assigned positions as vertex-disjoint cycles
in the compatibility graph. For
\[
    h=K\log(m+2),
\]
the random graph satisfies, with high probability, a uniform
connectivity property: every sufficiently large pair of admissible
vertex sets contains the arcs needed for a two-cycle, and every
sufficiently large family of first-stage paths can be completed using
a sufficiently large set of third vertices. Hall's theorem then
realizes all but $O(h)$ copies of each pattern. Since there are only
constantly many patterns, the total unrealized utility is
$O(\log m)$. 

\begin{proof}{Proof of Theorem~\ref{thm:cycle-random}.}
Let $m:=|V|$. By Lemma~\ref{thm:ApplyScarf}, there is a fractional
exchange $y^*$ such that it is sufficient to construct an augmented
integral exchange giving every organization $i$ at least
\(
    b_i:=
    \left\lfloor
        \sum_{c\in\mathcal C_3}\gamma_c^i y_c^*
    \right\rfloor
\)
transplants.

On the high-probability event in
Lemma~\ref{lem:logarithmic-three-rounding}, the lemma may be applied to
the graph-dependent exchange $y^*$ because its conclusion holds
simultaneously for every fractional exchange. It gives an integral
exchange $\Ex$, using only the original vertices, for which
\[
    \sum_{i\in\N}
    \bigl(b_i-u_i(\Ex)\bigr)^+
    \le C_0\log(m+2).
\]
For each organization $i$, define its remaining deficit by
\(
    e_i:=\bigl(b_i-u_i(\Ex)\bigr)^+.
\)

Because $b_i\le |U^i|$, organization $i$ has at least $e_i$ uncovered
utility-bearing patients. Add $e_i$ altruistic donors and assign them
to $e_i$ distinct uncovered patients of organization $i$. Under the
dummy-patient representation, each new altruistic donor and the
vertex associated with the uncovered patient form a directed
two-cycle. The resulting augmented exchange $\Ex^+$ therefore
satisfies
\[
    u_i(\Ex^+)
    \ge u_i(\Ex)+e_i
    \ge b_i
    \qquad(i\in\N).
\]
Lemma~\ref{thm:ApplyScarf} then implies that $\Ex^+$ belongs to the
supplemented core. Consequently, on the high-probability event,
\[
    d(\graph)
    \le \sum_{i\in\N}e_i
    \le C_0\log(m+2).
\]
After increasing the constant if necessary, this proves the stated
high-probability bound.

For every realization of the compatibility graph, the trivial bound
\(
    d(\graph)\le |U|\le m
\)
holds, since one additional altruistic donor can be assigned to every
utility-bearing patient. Hence,
\begin{align*}
    \mathbb E[d(\graph)]
    &\le
    C_0\log(m+2)
    +m\Pr\left(
        d(\graph)>C_0\log(m+2)
    \right)\\
    &\le
    C_0\log(m+2)+m(m+2)^{-4}\\
    &\le C\log(m+2)
\end{align*}
for a constant $C=C(g,p)$.

Finally, if the resulting supplemented-core exchange is not
Pareto-optimal, replace it by a Pareto-dominating exchange that
maximizes total utility in the same augmented compatibility graph.
Any coalition blocking the new exchange would also block the original
one. The replacement therefore remains in the supplemented core and
is Pareto-optimal.
\qed{}
\end{proof}

\paragraph{Extension to longer cycles.}
The same construction extends to every fixed cycle-length bound
$\Delta$. The number of augmented patterns remains constant, with the
constant now depending on $\Delta$, and a pattern can be realized by
extending vertex-disjoint partial paths one slot at a time. Each
matching stage loses at most $O(\log m)$ paths, and there are only
constantly many stages and patterns. Thus, for every fixed $\Delta$,
the required supplementation remains
$\mathcal O_{g,p,\Delta}(\log(m+2))$ in expectation and with
probability at least $1-(m+2)^{-4}$. The complete argument is given in
Appendix~\ref{app:fixed-cycle-length}.

\section{Numerical Results}\label{sec:simulation}
In this section, we study three experimental questions.  First, can the
platform find a maximum-cardinality exchange in the weak core?  Second, does a
representative operational allocation rule select such an exchange?  Third,
if the selected exchange is blocked, can the platform restore stability
without lowering any operational objective?  The distinction between the
first two questions matters: a donor-free weak-core exchange may exist even
when the operational rule selects a different, blocked exchange.


For supplementary robustness, we also study the strong and TU cores for coalitions of at most nine organizations.  This cap keeps the analysis computationally feasible and focuses the robustness check on deviations that require coordination among relatively few organizations.  The TU core is the strongest of the three notions used in the numerical analysis.

\begin{definition}
Given an exchange $\Ex$ and a coalition
$P\subseteq N=\{1,\ldots,n\}$, we say that $P$ is a \emph{TU blocking
coalition of $\Ex$} if there exists an exchange $\Ex'$ in the graph
$\mathcal{G}[\cup_{i\in P}V^i]$ such that
\[
    \sum_{i\in P}u_i(\Ex')>\sum_{i\in P}u_i(\Ex).
\]
We say that $\Ex$ is in the \emph{TU core} if it has no TU blocking
coalition.
\end{definition}

Thus, a coalition forms a TU block when it can increase its total number of
transplants, even if some coalition members receive fewer transplants.  The
TU core implies the strong core, and the strong core implies the weak core.

\begin{definition}
Given additional altruistic donors $V^0$, an exchange $\Ex^*$ in the extended
graph $\GG$ is in the \emph{$V^0$-supplemented TU core} if it is not TU
blocked by any coalition $P\subseteq N$.  If $|V^0|\leq d$, we also refer to
$\Ex^*$ as a \emph{$d$-supplemented TU-core exchange}.
\end{definition}

\subsection{Algorithms Under Study}
\label{subsec:simulation-algorithms}

\paragraph{Cycles, chains, and altruists.}
At each graph state, we enumerate every feasible directed 2- and 3-cycle.  An
external altruist is represented by a dummy patient vertex.  The vertex has
the donor's outgoing compatibility arcs and artificial return arcs
from real vertices to the dummy patient.  A cycle containing the altruist
therefore represents a donor-started chain.  Vertex-disjointness constraints
cover both real and altruist vertices, so the same donor cannot be used in
two selected chains.  The cycle database is rebuilt after each donor is
introduced.  An external donor with no compatible patient in the sampled pool
remains isolated; the procedure does not create synthetic compatibility arcs.

We distinguish donors \emph{introduced} from donors \emph{used}.  The
introduced prefix counts the donors made available to the market before the
procedure finds a certified stable allocation.  Donors used are those that
appear in the selected exchange.  The introduced prefix is the main
operational measure: if the first donor is not selected but stability is first
certified after a second donor is made available, the prefix is two even when
the final exchange uses only one donor.

For every cycle $c$, we record its number of real-patient transplants, its
incidence count $\gamma_c^i$ for each organization $i$, its number of
same-blood-type arcs, and its recipient-indegree score.  These values are used
in the integer programs below.

\paragraph{Exact weak-core separation.}
All primary procedures use the same exact mixed-integer separator.  Consider
a candidate allocation with utility $b_i$ for organization $i$.  Binary
$z_i$ indicates whether organization $i$ joins a deviating coalition, and
binary $y_c$ indicates whether donor-free cycle $c$ is used in the deviation.
For a coalition-size cap $\kappa$, the separator imposes
\[
  1\leq \sum_i z_i\leq\kappa,\qquad
  \sum_{c\ni v}y_c\leq 1\quad\forall v,\qquad
  y_c\leq z_i
\]
for every organization $i$ that owns a real vertex in cycle $c$.  Let
\[
      d_i=\sum_c\gamma_c^i y_c
\]
be the utility that organization $i$ receives in the deviation.  Weak-core
separation requires
\[
      d_i\geq (b_i+1)z_i \qquad\forall i
\]
and minimizes coalition size.  Any feasible solution encodes a coalition that
makes every member strictly better off.  Certified infeasibility is therefore
an exact stability certificate over the tested coalition set.  An interrupted
solve is recorded as inconclusive, rather than as stable.

\paragraph{Weak-core search.}
We first compute the donor-free maximum number $K$ of real-patient
transplants.  The search master retains real coverage of at least $K$.  When
the exact separator finds a blocking coalition $S$ at iteration $t$, with
current utilities $b_i^t$, the procedure adds the aggregate cut
\begin{equation}
  \sum_c\left(\sum_{i\in S}\gamma_c^i\right)x_c
  \;\geq\;
  \sum_{i\in S} b_i^t+1,
  \label{eq:weak-search-cut}
\end{equation}
where $x_c$ selects cycles in the next master solution.  The next donor in a
prespecified order is introduced only if the cut-augmented master becomes
infeasible.  Every returned allocation passes a final exact weak-core
separator.

Cut~\eqref{eq:weak-search-cut} is a useful search constraint, but it is not a
complete description of weak-core feasibility: it can exclude another stable
allocation with a different distribution of utility inside $S$.  Therefore,
a certified returned allocation is valid existence evidence, while search
failure or donor introduction does not prove that the donor-free weak core is
empty.   For $n\in\{5,10\}$, the cap of nine covers every proper coalition, so we reuse
donor-free TU-core certificates for the weaker notions and donor-free strong-core
certificates for the weak core. For $n\in\{20,30\}$, the strong- and TU-core
analyses are supplementary only, and weak-core stability is checked separately
against every proper coalition.


\paragraph{Donor-free lexicographic allocation rule.}
Our operational benchmark solves four exact integer-programming tiers:
\begin{enumerate}
    \item maximize the number of real-patient transplants;
    \item subject to tier 1, maximize the number of selected cycles, thereby
          favoring shorter exchanges;
    \item subject to tiers 1--2, maximize the number of same-blood-type arcs;
          and
    \item subject to tiers 1--3, maximize the sum, over selected cycles, of
          each cycle's maximum reciprocal recipient indegree.
\end{enumerate}
The first three tier values have integer coefficients and are normalized to
their certified integer values before they are fixed, removing only numerical
solver noise.  The recipient-indegree score in tier 4 is computed on the
donor-free graph.  This hierarchy represents
objectives used by kidney exchange programs
\citep{biro2021modelling}, although objectives and priorities differ across
programs.  The lexicographic rule does not impose a core constraint.  After
solving it, we audit the selected allocation with the exact weak-core
separator.

\paragraph{Floor-preserving lexicographic stabilization.}
Suppose that the donor-free lexicographic allocation is blocked.  Let
$F_\ell^*$ be its certified objective value at tier $\ell=1,\ldots,4$, and
let $F_\ell(x)$ be the same measure for allocation $x$.  Rather than
re-optimizing the four objectives after a donor is added, we retain all four
donor-free values as lower bounds and search for a stable allocation.  For a
given donor prefix, the master has the form
\begin{equation}
\begin{aligned}
 \min_x\quad & \sum_c a_c x_c \\
 \text{s.t.}\quad
 & \sum_{c\ni v}x_c\leq 1 && \forall v, \\
 & F_\ell(x)\geq F_\ell^* && \ell=1,\ldots,4, \\
 & \sum_c\left(\sum_{i\in S}\gamma_c^i\right)x_c
      \geq \sum_{i\in S}b_i^t+1
      && \forall(S,t)\in\mathcal{K}, \\
 & x_c\in\{0,1\} && \forall c.
\end{aligned}
 \label{eq:lex-floor-master}
\end{equation}
Here, $a_c$ is the number of altruists in cycle $c$, and $\mathcal{K}$ is the
set of aggregate weak-core cuts collected during the search.  The master
minimizes donors used while preserving or improving every donor-free
operational benchmark.  Recipient-indegree scores remain fixed at their
donor-free values, so tier 4 has the same meaning after donor arcs are added.

The exact separator audits each master solution.  A new donor is introduced
only after the current cut-augmented master is certified infeasible.  This
allows the procedure first to search for a different tied donor-free
allocation.  Conditional on an initial block, we test up to three seeded donor
orders.  We stop early if the first order yields a certified donor-free
allocation.  A successful allocation is exactly verified, but the aggregate
cuts make the search heuristic for existence.  We therefore report the
order-specific introduced prefixes and two market-level summaries: the
smallest certified successful prefix among the tested orders, called the
\emph{best tested prefix}, and the largest certified successful prefix, called
the \emph{worst certified tested prefix}.  The latter is an order/seed
sensitivity measure, not an adversarial worst-case guarantee.  Neither summary
is a proof of a global minimum or maximum donor requirement.

\paragraph{Supplementary strong- and TU-core procedures.}
The strong-core search uses the same coverage floor, sequential donor
introduction, and aggregate cut~\eqref{eq:weak-search-cut}, with an exact
strong-core separator.  A positive result is certified, but the accumulated
cuts make a negative result heuristic.  The exact TU procedure instead adds
the valid coalition constraint
\[
       \sum_{i\in S}u_i(\Ex)\geq v(S),
\]
where $v(S)$ is the maximum total utility attainable by coalition $S$ in its
donor-free induced graph.  The TU master minimizes donors used subject to
coverage of at least $K$ and all accumulated TU constraints.  Certified
infeasibility before any donor is introduced proves donor-free TU-core
nonexistence over the tested coalition set.  These supplementary procedures
test coalitions of at most nine organizations for computational feasibility.  This is a full-core check for
$n=5$ and $n=10$; for $n=20$ and $n=30$, it is a bounded-coalition robustness
check.

\subsection{Simulation Framework}
\label{subsec:simulation-framework}

We generate partition-exchange economies using the generator of
\citet{delorme2022improved} and \citet{pettersson2021kidney}, which refines
the Saidman generator \citep{saidman2006increasing}.  This design is
consistent with computational kidney-exchange studies
\citep{benedek2023computing,benedek2024computing,druzsin2024performance}.
The design varies the following quantities.
\begin{itemize}
    \item \textbf{Base instances.}  We use a seeded sample of 30 independent
 base instances.  Each contains 1,000 patient--donor pairs and 53
    simulated external altruists.  The base instances are the independent
    sampling clusters.

    \item \textbf{Pool size.}  From every base instance, we take independently
    seeded induced samples of 100, 200, and 500 patient--donor pairs.  The
    three samples are not nested.

    \item \textbf{Organizations and partitions.}  We use
    $n\in\{5,10,20,30\}$.  To reflect the different sizes of participating organizations, we draw organization-size shares from a symmetric
    Dirichlet distribution with concentration parameter 5 and randomly
    shuffle the vertices before assigning them to organizations.  We draw two
    partitions for every graph and organization count.

    \item \textbf{Cycle-length cap.}  We use $\Delta\in\{2,3\}$.

    \item \textbf{Coalitions.}  The primary weak-core and lexicographic
    procedures test every proper coalition and therefore provide full-core
    evidence.  The supplementary TU and strong procedures use a cap of nine
    for computational feasibility.

    \item \textbf{Donor order.}  Conditional on an algorithm reaching a case
    in which order can matter, we test up to three seeded random orders.  The
    same orders are used across procedures and cycle caps.  These are repeated
    trials within a market, not independent observations.  The order
    repetitions also use different algorithm and solver seeds, so we refer to
    their variation as order/seed sensitivity.
\end{itemize}

The primary design contains
$30\times3\times4\times2\times2=1{,}440$ markets, with 60 markets in every
pool-size, organization-count, and cycle-cap cell.  

The experiments were run on an AMD Ryzen 7735HS processor with 16GB RAM using
Gurobi through PuLP.  Standard primary exact MIP calls used a 300-second limit.  The pipeline automatically
reran six hard lexicographic markets with a 1,200-second limit.  These are per-MIP limits, so a cutting-plane procedure containing
several solves may run longer than its stated limit.

\subsection{Results}
\label{subsec:simulation-results}

\begin{table}[t]
\centering
\small
\setlength{\tabcolsep}{5pt}
\begin{tabular}{rrrrrrr}
\toprule
Organizations & \multicolumn{2}{c}{Weak-core procedure} &
\multicolumn{2}{c}{Initial lexicographic} &
\multicolumn{2}{c}{Final lexicographic} \\
\cmidrule(lr){2-3}\cmidrule(lr){4-5}\cmidrule(lr){6-7}
$n$ & Donor-free & Assisted & Stable & Blocked & Donor-free & Assisted \\
\midrule
5  & 360 & 0 & 344 & 16 & 346 & 14 \\
10 & 360 & 0 & 332 & 28 & 336 & 24 \\
20 & 360 & 0 & 346 & 14 & 350 & 10 \\
30 & 359 & 1 & 347 & 13 & 349 & 11 \\
\midrule
Total & 1,439 & 1 & 1,369 & 71 & 1,381 & 59 \\
\bottomrule
\end{tabular}
\caption{Full weak-core outcomes by number of organizations.  Each row
contains 360 markets.  \emph{Assisted} means that the procedure returned a
certified weak-core allocation after a positive donor prefix was introduced;
it counts markets, not donors, and does not prove that assistance was
necessary.  Final lexicographic outcomes preserve all four donor-free
objective levels.}
\label{tab:full-core-outcomes}
\end{table}

\paragraph{Weak-core existence and the selection gap.}
Table~\ref{tab:full-core-outcomes} shows that the weak-core procedure returns
a certified full-core allocation in every market.  It finds a donor-free
allocation in 1,439 markets, or 99.93\%, and a verified allocation in the
remaining market; the worst tested order introduces two donors there.  Because
the weak search uses aggregate cuts, this last result does not prove that the
donor-free weak core is empty.

The donor-free lexicographic rule is stable in 1,369 markets and blocked in
71.  The overall blocking rate is 4.93\%, with a base-instance-clustered 95\%
confidence interval of [3.82\%, 6.04\%].  Seventy of these 71 markets also
have a donor-free weak-core witness from the core-aware procedure.  Thus,
almost every observed lexicographic block is directly shown to be a selection
problem rather than a lack of positive evidence for weak-core existence.

The floor search finds a different donor-free allocation in 12 of the 71
blocked markets.  It introduces a positive donor prefix in the other 59, an
overall rate of 4.10\% [3.13\%, 5.07\%].  Every returned allocation passes
the exact full weak-core separator and meets all four donor-free
lexicographic floors.  Hence stability is restored without sacrificing any
recorded operational objective.

\subsubsection{Number of Donors Introduced}
\label{subsubsec:number-donors-introduced}

The introduced prefix is the relevant operational count: it records how many
donors must be considered and made available before the implemented procedure
finds a certified stable allocation.  If the first introduced donor is not
used and stability is first found after the second donor is introduced, the
prefix is two.  This remains true even if only the second donor appears in the
selected exchange.

\begin{table}[t]
\centering
\small
\setlength{\tabcolsep}{4.5pt}
\begin{tabular}{llrrr}
\toprule
Summary of introduced prefixes & Unit & One donor & Two donors & Three donors \\
\midrule
Best tested order & 59 assisted markets & 58 & 1 & 0 \\
All certified order-specific outcomes & 173 order runs & 160 & 9 & 4 \\
Worst certified tested order & 59 assisted markets & 48 & 7 & 4 \\
\bottomrule
\end{tabular}
\caption{Numbers of donors introduced by floor-preserving lexicographic
stabilization.  The first and third rows take the best and worst certified
successful prefix, respectively, within each market.  The middle row reports
every certified successful order-specific run.  Donor orders are repeated
trials within markets, not independent observations.}
\label{tab:lex-introduced-donors}
\end{table}

Table~\ref{tab:lex-introduced-donors} reports the best and worst certified
tested results within each market as well as every certified order-specific run.
Under the best of up to three orders, the procedure introduces one donor in 58
of the 59 assisted markets and two donors in one, for a mean of 1.02 donors.
Across all 173 certified order-specific outcomes, it introduces one donor in
160 runs, two in nine, and three in four; the mean is 1.10 and the median is
one.  More conservatively, under the worst certified tested order in each
market, the prefix is one in 48 markets, two in seven, and three in four, for a
mean of 1.25 donors and a median of one.  Thus the average prefix rises only
modestly when we move from the most favorable to the least favorable certified
tested order, and no certified order-specific result introduces more than three
donors.  The only market whose best tested prefix is two has 100 pairs, 30
organizations, and $\Delta=2$.  All four order-specific prefixes of three also
occur in 100-pair markets with 20 or 30 organizations.

The number of donors used in the exchange is smaller in a few cases, but it
is a secondary measure.  Among the best tested outcomes, every assisted
exchange uses one donor.  The sole donor-assisted weak-core market has
introduced prefixes of one, two, and one under its three tested orders, and
each returned exchange uses one donor.  Because the searches use aggregate
cuts, these introduced counts describe the implemented algorithm under the
tested orders; they are not proofs of the theoretical minimum number of
donors.

\begin{figure}[t]
    \centering
    \includegraphics[width=\linewidth]{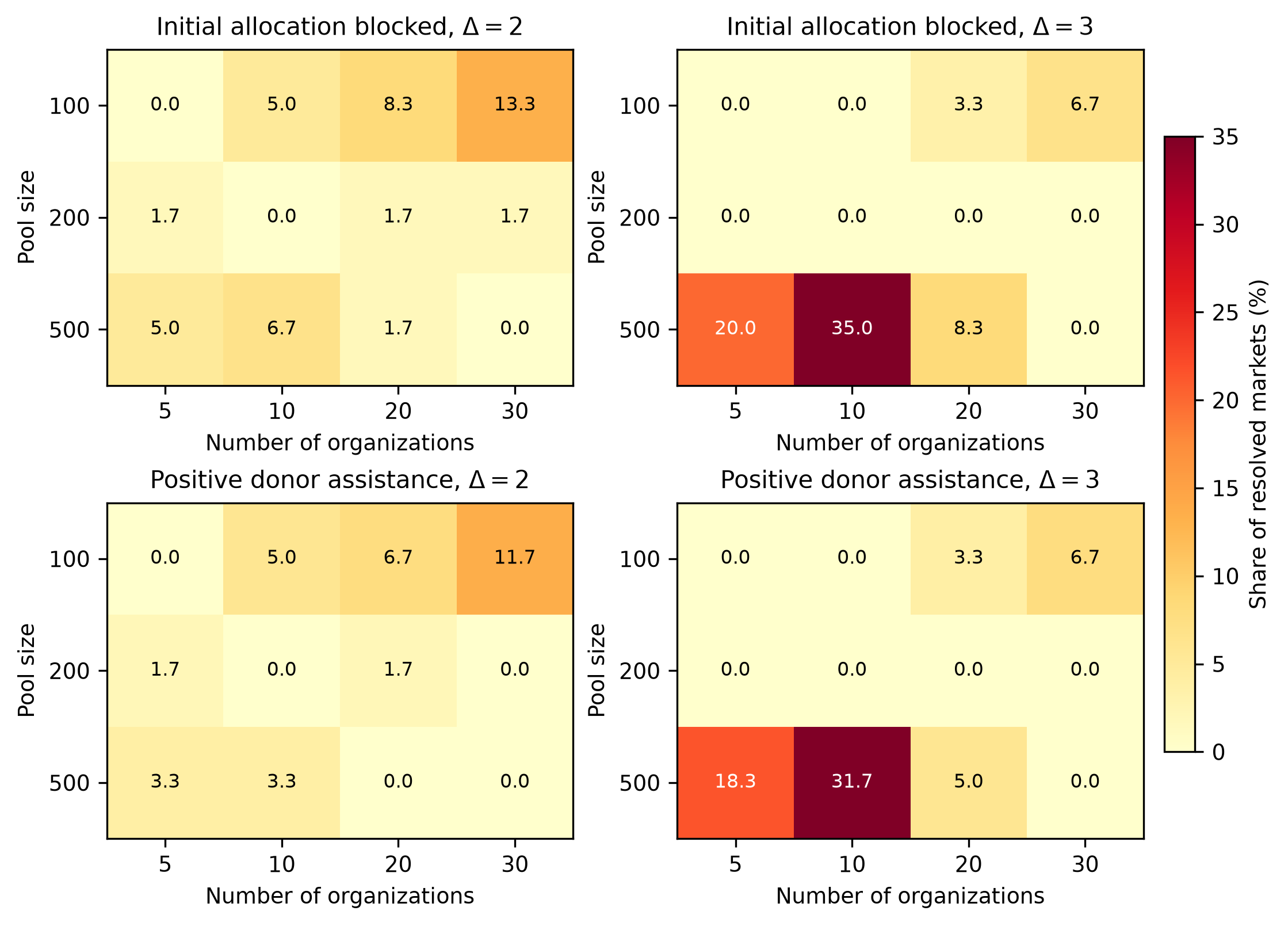}
    \caption{Full-core lexicographic outcomes by pool size, number of
    organizations, and cycle-length cap.  The top panels report the share of
    donor-free lexicographic allocations with a certified weak block.  The
    bottom panels report the share of final outcomes using positive donor
    assistance.  Every cell contains 60 markets.}
    \label{fig:full-core-parameter-effects}
\end{figure}


\subsubsection{Pool Size, Organization Count and Introduced Donors}
\label{subsubsec:pool-donor-assistance}

\begin{table}[t]
\centering
\small
\begin{tabular}{rrrrrr}
\toprule
Pool size & Markets & Initial blocks & Block rate & Assisted outcomes &
Assistance rate \\
\midrule
100 & 480 & 22 & 4.58\% & 20 & 4.17\% \\
200 & 480 & 3  & 0.63\% & 2  & 0.42\% \\
500 & 480 & 46 & 9.58\% & 37 & 7.71\% \\
\bottomrule
\end{tabular}
\caption{Pool size and floor-preserving lexicographic stabilization.  Rates
use all 480 markets at each pool size; all initial and final outcomes are
certified.}
\label{tab:pool-size-assistance}
\end{table}

Table~\ref{tab:pool-size-assistance} and
Figure~\ref{fig:pool-size-donor-assistance} show a U-shaped relation between
pool size and the frequency of assistance.  The 200-pair pools have the
lowest assistance rate.  The 500-pair pools have the highest rate, but this
is concentrated in markets with five or ten organizations and $\Delta=3$.
The 100-pair assistance cases are concentrated in markets with 20 or 30
organizations.  

The heatmap suggests two distinct regimes in which donor assistance is most
likely to be useful.  In small pools, assistance is concentrated when the
market is highly fragmented across many organizations, so that each
organization controls only a relatively small part of an already thin pool.
In large pools, by contrast, assistance is concentrated when relatively few
organizations divide a much thicker market.  Thus, instability does not arise
simply from markets being ``small'' or ``large.''  Rather, it appears at two
different extremes of the relationship between pool thickness and
organizational fragmentation: thin and highly fragmented markets, and thick
markets divided among relatively few organizations.  The intermediate
200-pair markets exhibit much less need for intervention in our simulations.

Pool size affects the frequency of intervention more than the length of the
introduced prefix.  Among the 59 assisted markets, the mean best tested
introduced prefix is 1.05 donors at pool size 100 and exactly 1.00 at pool
sizes 200 and 500; the median is one at every pool size.  Across all certified
order-specific outcomes, the corresponding means are 1.23, 1.17, and 1.02,
and the maxima are three, two, and two.  Thus, the higher assistance rate in
some 500-pair cells does not translate into a longer donor prefix once an
intervention is made.

Allowing 3-cycles also has a conditional effect.  Across all markets, initial
blocking rises from 27 of 720 markets (3.75\%) at $\Delta=2$ to 44 of 720
(6.11\%) at $\Delta=3$.  Positive assistance rises from 20 (2.78\%) to 39
(5.42\%). This aggregate increase comes mainly from the
500-pair pools with five and ten organizations.  In the 100-pair pools,
3-cycles instead reduce blocking for 10, 20, and 30 organizations.  

\begin{figure}[t]
    \centering
    \includegraphics[width=\linewidth]{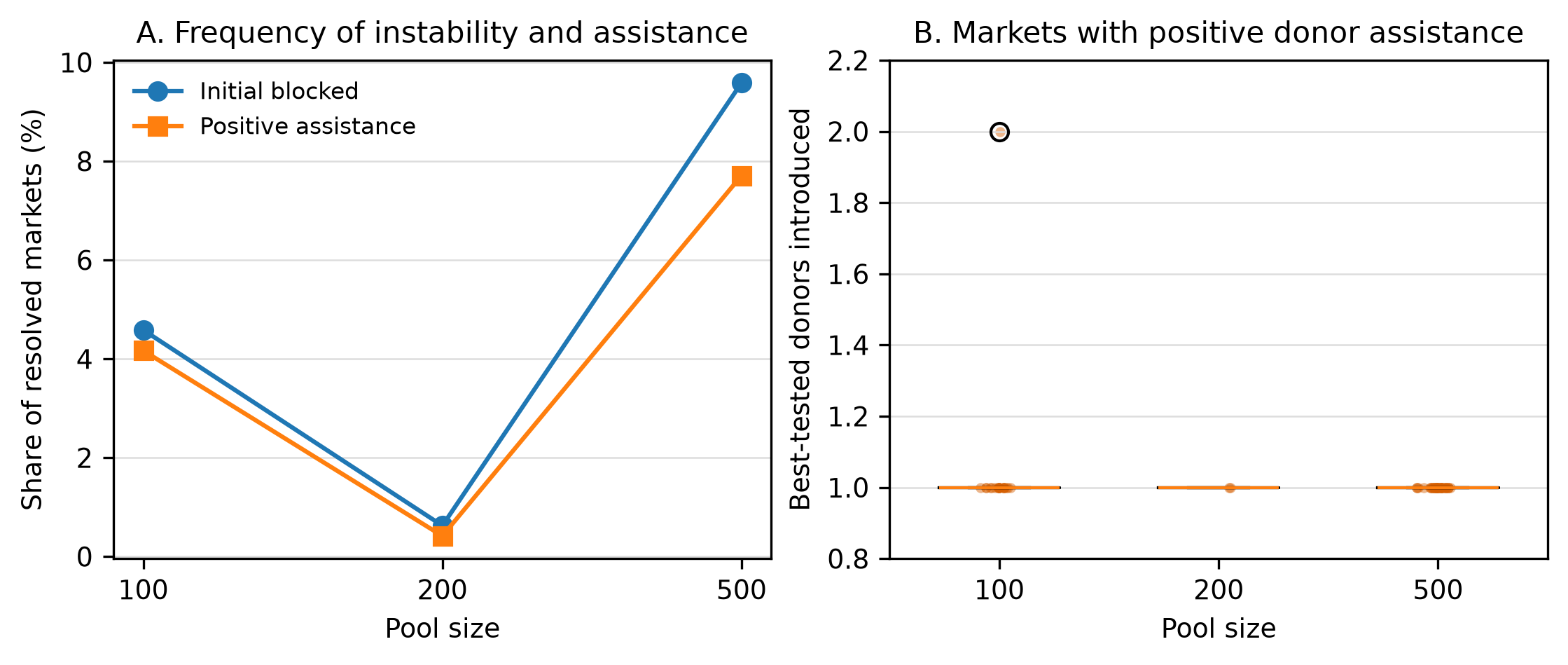}
    \caption{Pool size and lexicographic stabilization.  Panel A reports the
    frequency of an initial block and positive donor assistance.  Panel B
    reports the best tested introduced prefix in the 59 markets with positive
    assistance.}
    \label{fig:pool-size-donor-assistance}
\end{figure}


\paragraph{Partition and base-instance sensitivity.}
The two partition draws give similar aggregate results: the numbers of
assisted lexicographic outcomes are 28 and 31.  Holding the graph and all
other design variables fixed, the two partitions disagree on the initial
block indicator in 55 of 720 paired settings and on positive assistance in
47 of 720.  Positive assistance appears in 25 of the 30 independent base
instances, and no base instance contributes more than five of the 59 cases.
The result is therefore not driven by one generated compatibility graph.

\begin{table}[t]
\centering
\small
\setlength{\tabcolsep}{4.5pt}
\begin{tabular}{llrrrr}
\toprule
Procedure & Organizations & Donor-free & Assisted & Unresolved & Mean best prefix \\
\midrule
Strong & 5  & 359 & 1   & 0 & 1.00 \\
Strong & 10 & 340 & 20  & 0 & 1.00 \\
Strong & 20 & 305 & 55  & 0 & 1.47 \\
Strong & 30 & 253 & 107 & 0 & 2.31 \\
\midrule
TU & 5  & 357 & 1   & 2 & 1.00 \\
TU & 10 & 320 & 37  & 3 & 1.14 \\
TU & 20 & 285 & 71  & 4 & 1.66 \\
TU & 30 & 245 & 115 & 0 & 2.61 \\
\bottomrule
\end{tabular}

\vspace{0.6em}
\begin{tabular}{lrrrr}
\toprule
Procedure & Assisted markets & Mean best prefix & Max best prefix & Largest tested prefix \\
\midrule
Strong & 183 & $\approx 1.91$ & 7 & 19 \\
TU     & 224 & 2.06 & 7 & 19 \\
\bottomrule
\end{tabular}
\caption{Supplementary cap-nine strong- and TU-core results and donor-order
sensitivity.  The upper panel reports market outcomes and the mean of the
smallest certified successful prefix among the tested orders.  The lower panel
aggregates the assisted markets and contrasts that optimistic benchmark with
the largest prefix observed in any certified order-specific run.  The latter is
a tested order/seed stress measure, not a worst-case guarantee over all donor
orders.  The strong search is heuristic; its assisted counts do not prove
donor-free nonexistence.  The checks are full-core for $n=5,10$ and
bounded-coalition robustness checks for $n=20,30$.}
\label{tab:supplementary-core-results}
\end{table}

\subsubsection{Supplementary strong- and TU-core robustness.}
Table~\ref{tab:supplementary-core-results} shows that the stronger notions
become more demanding as the number of organizations rises.  The strong
procedure returns a certified allocation in every market.  Across its 183
assisted markets, the mean best tested prefix is approximately 1.91 donors and
the largest best tested prefix is seven.  The exact TU procedure records 1,207
donor-free outcomes, 224 donor-assisted outcomes, and nine unresolved markets;
its mean best tested prefix is 2.06 donors and its largest best tested prefix is
also seven.  The unresolved TU cases are kept as such and are not classified as
unstable.

The best-order statistics alone, however, understate sensitivity to donor
arrival order.  Among certified order-specific runs, the largest introduced
prefix is 19 for the strong-core heuristic and also 19 for the TU-core procedure.
Thus, unlike the primary lexicographic stabilization result, for which even the
worst certified tested order never exceeds three donors, the stronger solution
concepts can require a materially larger reserve under an unfavorable tested
order.  We view this as useful stress-test evidence rather than as a negative
result: most assisted markets are resolved with much smaller prefixes, but a
platform that wants robustness to the stronger notions should plan for more
order sensitivity than is suggested by the best-tested prefix alone.  Because
only up to three seeded orders are tested and the repetitions also change
algorithm and solver seeds, these maxima are not adversarial worst-case donor
requirements.

For 20 and 30 organizations, these results concern coalitions of at most nine
organizations and are supplementary to the full weak-core study.  In addition,
strong-core donor counts are upper bounds returned by a heuristic search.  The
main full-core conclusion therefore rests on the primary weak-core and
lexicographic procedures; the strong and TU results show how the reserve can
grow when both the stability concept and donor-order stress are strengthened.

\begin{table}[t]
\centering
\small
\setlength{\tabcolsep}{4pt}
\begin{tabular}{lrrrrrr}
\toprule
Procedure & Calls & Certified & Median (s) & 90th pct. (s) & Maximum (s) &
Sum of call times (h) \\
\midrule
Weak heuristic & 749 & 744 & 0.107 & 21.4 & 1,597.2 & 2.28 \\
Initial lexicographic audit & 1,440 & 1,440 & 0.207 & 27.4 & 228.9 & 2.86 \\
Floor-preserving stabilization & 189 & 185 & 44.64 & 677.6 & 2,022.5 & 10.38 \\
\bottomrule
\end{tabular}
\caption{Recorded elapsed time per actual procedure call.  The final column
is the sum of call times, not the wall-clock duration of the study.  Weak
allocations inferred from TU or strong results are excluded because the weak
heuristic was not run.  Floor calls are conditional on an initial block.  A
market can be resolved even when one order-specific call is inconclusive.}
\label{tab:full-core-runtimes}
\end{table}

\begin{figure}[t]
    \centering
    \includegraphics[width=\linewidth]{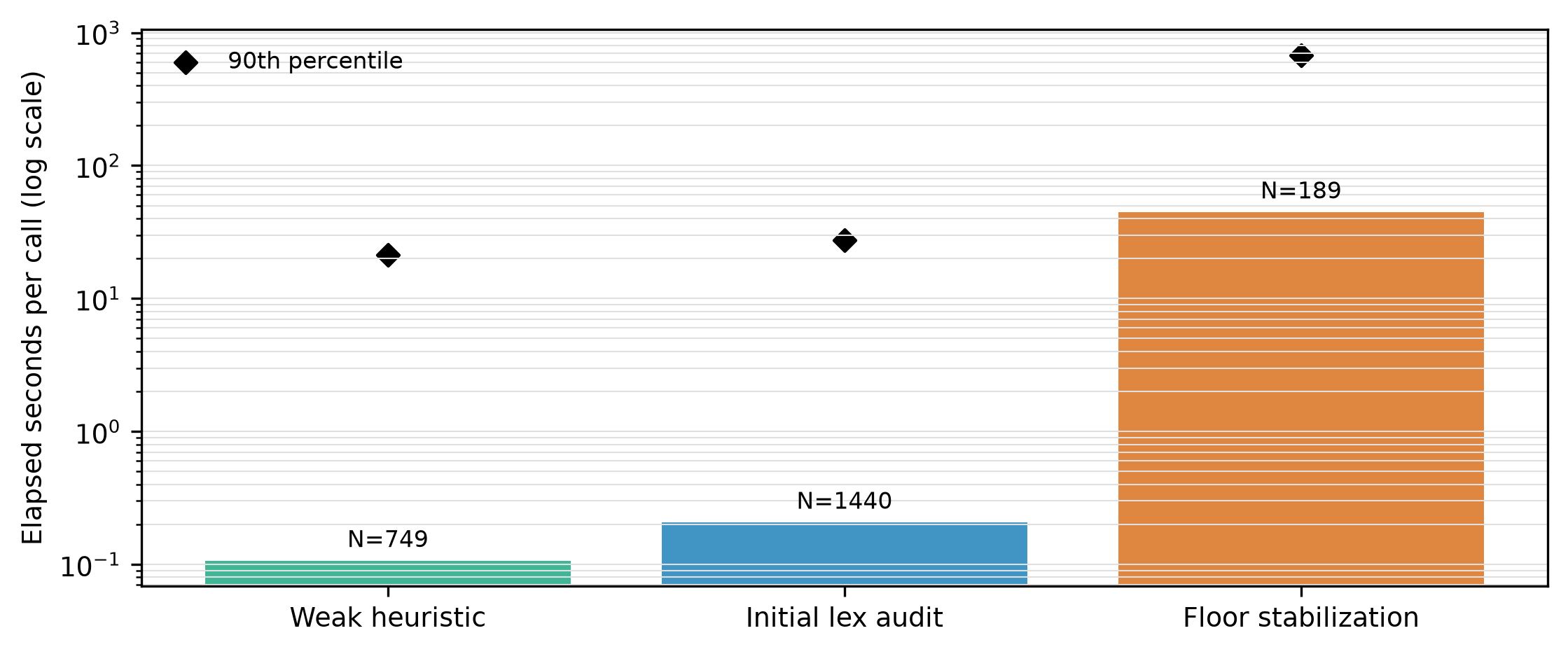}
    \caption{Median and 90th-percentile elapsed time per actual procedure
    call.  The vertical scale is logarithmic.  Calls are conditional, and the
    solver limit applies to each inner MIP rather than to the complete
    procedure.}
    \label{fig:full-core-runtime-profiles}
\end{figure}

\paragraph{Runtime.}
Table~\ref{tab:full-core-runtimes} and
Figure~\ref{fig:full-core-runtime-profiles} show strongly right-skewed
runtimes.  Weak calls have a median of 0.107 seconds, while the 90th
percentile is 21.4 seconds.  Initial lexicographic audits have a median of
0.207 seconds.  Conditional floor searches are harder, with a median of 44.64
seconds and a 90th percentile of 677.6 seconds.  Almost all recorded time in
all three groups comes from 500-pair markets with $\Delta=3$.

\subsubsection{Summary.}
The simulations give three main findings.  First, the weak-core procedure finds
a certified full-core exchange in every market and does so without positive
donor assistance in 1,439 of 1,440 markets; in the remaining market the three
tested donor orders introduce one, two, and one donors, respectively.  Second,
the operational lexicographic rule sometimes selects a blocked allocation even
when another donor-free weak-core allocation has been found.  Floor-preserving
search closes this selection gap in every tested market without lowering any of
the four donor-free objective levels.  Positive donor prefixes are introduced
in 59 markets.  The mean prefix is 1.02 under the best tested order and 1.25
under the worst certified tested order; the corresponding worst-order prefixes
are one in 48 markets, two in seven, and three in four.  Thus the conclusion
that only a small reserve is needed for the primary implementation result is
robust to the tested donor-order/seed variation. Third, the markets that need donor assistance for core stability are generally either small and very fragmented, or large with only a few, large organizations.

The stronger-core robustness
checks are more order-sensitive: although their best-tested prefixes remain
small on average, individual certified tested orders reach 19 donors for the
strong-core heuristic and also 19 for the TU-core procedure.  These are sensitivity
diagnostics rather than adversarial worst-case guarantees.  The frequency of
intervention depends on pool size, organizational fragmentation, and cycle
length through strong interactions, whereas the primary weak-core
implementation result remains quantitatively small across the tested orders.

\section{Applications to One-Sided Exchange Economies}\label{sec:bipartite}

We next apply the bipartite case of Theorem~\ref{thm:arbitrary} to an
indivisible-goods exchange economy. The result extends the core-existence
theorem of \citet{echenique2024stable}. The proof uses an auxiliary
partition exchange economy in which goods are represented by altruistic-donor
vertices and the positions in an assignment valuation are represented by
patient--donor-pair vertices. We make the endowment and blocking
correspondence explicit below.

There is a finite set of goods
\[
    M=M^1\mathbin{\dot\cup}\cdots\mathbin{\dot\cup}M^n,
\]
where \(M^i\) is the endowment of agent \(i\in N\). Agent \(i\) has a
valuation \(v_i:2^M\to\mathbb{R}\). An allocation is a partition
\(M=M_1^*\mathbin{\dot\cup}\cdots\mathbin{\dot\cup}M_n^*\).

\begin{definition}
An allocation \(M^*=(M_1^*,\ldots,M_n^*)\) is in the weak core if there
do not exist a coalition \(S\subseteq N\) and a partition
\[
    \bigcup_{i\in S}M^i
       =\mathop{\dot\bigcup}_{i\in S}M_i'
\]
such that \(v_i(M_i')>v_i(M_i^*)\) for every \(i\in S\).
\end{definition}

\citet{echenique2024stable} consider categorical economies. In such an
economy, the goods are partitioned into categories,
\(M=O^1\mathbin{\dot\cup}\cdots\mathbin{\dot\cup}O^K\), and each agent
\(i\) has a set \(G_i\subseteq M\) of acceptable goods. Its valuation is
\[
    v_i(X)=\sum_{k=1}^K
       \min\bigl\{|X\cap O^k\cap G_i|,1\bigr\}.
\]
Thus, an agent obtains one unit of utility from an acceptable good in a
category and obtains no additional utility from receiving a second acceptable
good in the same category. These preferences are additively separable and
dichotomous.

\begin{theorem}[\citet{echenique2024stable}]\label{theo:federico}
A categorical economy with additively separable and dichotomous preferences
has a nonempty weak core.
\end{theorem}

We establish core existence for the larger class of binary assignment
valuations.

\begin{definition}[Binary assignment valuation]
A valuation \(v:2^M\to\mathbb{R}\) is a \emph{binary assignment
valuation} if there is a finite set of positions \(J\) and a matrix
\(\alpha\in\{0,1\}^{M\times J}\) such that, for every \(X\subseteq M\),
\[
v(X)=\max_{z}\left\{
       \sum_{g\in X}\sum_{j\in J}\alpha_{gj}z_{gj}:
       \begin{array}{ll}
       \displaystyle\sum_{j\in J}z_{gj}\leq 1
           & \text{for every }g\in X,\\[3pt]
       \displaystyle\sum_{g\in X}z_{gj}\leq 1
           & \text{for every }j\in J,\\[3pt]
       z_{gj}\in\{0,1\}
           & \text{for every }g\in X,\ j\in J
       \end{array}
       \right\}.
\]
Equivalently, \(v(X)\) is the size of a maximum matching between the goods
in \(X\) and the positions in \(J\), where \(g\) can be assigned to \(j\)
if and only if \(\alpha_{gj}=1\).
\end{definition}

The restriction \(\alpha_{gj}\in\{0,1\}\) distinguishes this class from the
more general assignment valuations studied by \citet{milgrom2009assignment}.
Every categorical valuation above is a binary assignment valuation. For agent
\(i\), take one position \(j_k\) for each category \(O^k\) and set
\[
    \alpha^i_{g j_k}=1
    \quad\Longleftrightarrow\quad
    g\in O^k\cap G_i.
\]
Matching goods to these positions gives exactly
\(\sum_k\min\{|X\cap O^k\cap G_i|,1\}\).

The inclusion is strict. Consider three goods \(a,b,c\) and the valuation
\(v(X)=\min\{|X|,2\}\). This is a binary assignment valuation: take two
positions and let every good be acceptable at both positions. It cannot be
represented by categorical preferences. Because each singleton has value one,
all three goods would have to be acceptable. Because every pair has value two,
each pair would then have to lie in different categories. The three goods
would therefore lie in three different categories, implying
\(v(\{a,b,c\})=3\), a contradiction.

\begin{theorem}\label{thm:binary-assignment-core}
If every agent has a binary assignment valuation, then the weak core is
nonempty.
\end{theorem}

\begin{proof}{Proof.}
For each agent \(i\), fix a representation of \(v_i\) by a set of positions
\(J^i\) and a binary matrix
\(\alpha^i\in\{0,1\}^{M\times J^i}\). We may take the sets \(J^i\) to be
pairwise disjoint. We construct an auxiliary partition exchange economy with
pairwise exchange.

For each good \(g\in M^i\), create an altruistic-donor vertex, also denoted by
\(g\), and assign it to organization~\(i\). For each position \(j\in J^i\),
create a patient--donor-pair vertex, also denoted by \(j\), and assign it to
the same organization. Thus, organization \(i\) owns
\[
    V^i=M^i\mathbin{\dot\cup}J^i.
\]
This is the endowment correspondence: a coalition \(S\) controls precisely
the goods in \(\bigcup_{i\in S}M^i\) and the positions in
\(\bigcup_{i\in S}J^i\). The altruistic-donor vertices here represent the
agents' endowed goods; they are not additional donors supplied by the
platform.

Let \(J=\bigcup_{i\in N}J^i\). The graph is bipartite, with vertex classes
\(M\) and \(J\). For \(g\in M\) and \(j\in J^i\), include the edge
\(\{g,j\}\) if and only if \(\alpha^i_{gj}=1\); include no other edges.
In the directed kidney-exchange representation, the donor at \(g\) is
compatible with the patient at \(j\) exactly when
\(\alpha^i_{gj}=1\), while the donor paired with \(j\) is compatible with
the dummy patient attached to \(g\). Hence, each edge \(\{g,j\}\) is a
feasible pairwise exchange. Matching \(g\) to \(j\in J^i\) has the economic
interpretation of assigning good \(g\) to agent \(i\) through position \(j\).

The graph is bipartite. Theorem~\ref{thm:arbitrary} therefore gives a
matching \(\mathcal E^*\) in the weak core of the auxiliary partition
exchange economy. We now convert this matching into an allocation of the
goods. For each agent \(i\), let
\[
    \widehat M_i
       =\bigl\{g\in M:\{g,j\}\in\mathcal E^*
                    \text{ for some }j\in J^i\bigr\}.
\]
Because \(\mathcal E^*\) is a matching, the sets \(\widehat M_i\) are
pairwise disjoint. Let
\(R=M\setminus\bigcup_{i\in N}\widehat M_i\) be the set of unmatched goods
and return each such good to its original owner. Define
\[
    M_i^*=\widehat M_i\cup(R\cap M^i), \qquad i\in N.
\]
The sets \(M_i^*\) form a partition of \(M\). Moreover, the edges of
\(\mathcal E^*\) incident to \(J^i\) give a feasible assignment of the goods
in \(\widehat M_i\) to distinct positions in \(J^i\). Consequently,
\[
    v_i(M_i^*)\ \geq\ |\widehat M_i|
       \ =\ u_i(\mathcal E^*).
    \tag{\(*\)}
\]
Equality is not needed: returning an unmatched good may raise an agent's
valuation, which only makes that agent harder to improve upon.

It remains to show that the construction preserves coalitional deviations.
Suppose, toward a contradiction, that a coalition \(S\subseteq N\) strongly
blocks \(M^*\). Then \(S\) can repartition its original endowment as
\[
    \bigcup_{i\in S}M^i
       =\mathop{\dot\bigcup}_{i\in S}M_i'
\]
so that \(v_i(M_i')>v_i(M_i^*)\) for every \(i\in S\). For each \(i\in S\),
choose an assignment attaining \(v_i(M_i')\). Retain the pairs
\((g,j)\) with \(\alpha^i_{gj}=1\). These pairs form a matching between
\(M_i'\) and \(J^i\) of size \(v_i(M_i')\). Because the bundles \(M_i'\)
are disjoint and the position sets \(J^i\) are disjoint, the union of these
matchings over \(i\in S\) is a feasible matching \(\mathcal E'\) in the
subgraph induced by
\[
    \bigcup_{i\in S}V^i
       =\left(\bigcup_{i\in S}M^i\right)
          \mathbin{\dot\cup}
        \left(\bigcup_{i\in S}J^i\right).
\]
For every \(i\in S\), this matching gives
\(
    u_i(\mathcal E')=v_i(M_i')
       >v_i(M_i^*)
       \geq u_i(\mathcal E^*),
\)
where the last inequality follows from \((*)\). Thus \(S\) blocks
\(\mathcal E^*\) in the auxiliary partition exchange economy, contradicting
the choice of \(\mathcal E^*\). Hence \(M^*\) is in the weak core.
\end{proof}

\section{Conclusions}\label{sec:conclude}

Core stability has generally been viewed as a demanding and often unattainable benchmark for kidney exchange among organizations. The core may be empty, and even when it is nonempty, finding and implementing a stable exchange can be difficult. Our results offer a more optimistic perspective. Exact core stability can often be achieved through a modest intervention: supplementing the exchange with altruistic donors. Rather than weakening the stability requirement or relying on monetary transfers, the platform expands the feasible set just enough to satisfy organizations’ participation constraints. The supplemented core therefore turns stability from a passive property of the market into an objective that the platform can actively pursue.

The theoretical and computational results provide complementary support for this approach. The worst-case bounds show that stability can be costly in arbitrary compatibility graphs. In realistic environments,  however, the required supplementation remains small relative to market size. The simulations go further: core-stable outcomes are usually available, and outcomes selected by practical lexicographic rules can be stabilized by introducing only a few realistically generated donors. Thus, the supplemented core is not merely an existence result. It provides a way to reconcile stability with the clinical and operational priorities already used by kidney exchange programs.

A practical next step is to translate the supplemented core into an operating protocol. A platform could maintain a reserve of voluntarily committed altruistic donors, test proposed exchanges for blocking coalitions, and deploy donors only when needed to restore stability. In a dynamic setting with repeated clearing, the reserve may be partly self-replenishing: an altruist-initiated chain can leave a bridge donor to initiate a later chain, so using an altruist today need not reduce future stabilizing capacity one-for-one. Further computational work should incorporate stability into existing allocation software and extend the guarantees to richer operational constraints.



\bibliographystyle{plainnat}
\bibliography{MS/MS_refs}

\newpage

\appendix
\begin{center}
\large{APPENDIX}    
\end{center}
\vspace{5mm}

\section{Omitted Proofs for Lower Bounds}

\subsection{Proof of Theorem~\ref{thm:general-empty-core}}

\begin{figure}[h!]
    \centering
    \includegraphics[width=0.8\linewidth]{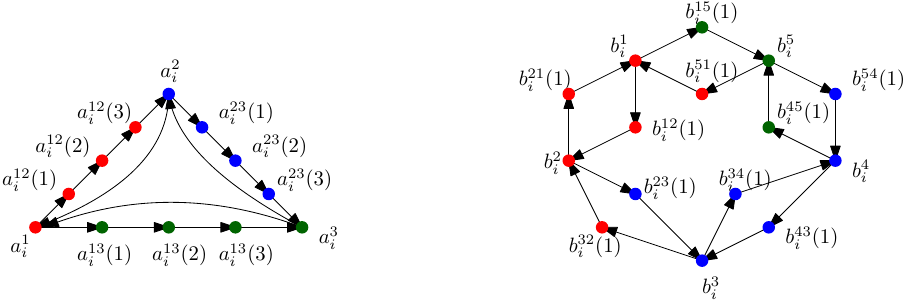}
    \caption{An illustration for Theorem~\ref{thm:general-empty-core} with $\Delta=5$. A type-A graph is on the left and a type-B graph is on the right. The index $i$ satisfies $i\equiv 0 \pmod 2$ and $i\equiv 1 \pmod 3$. Player 1 owns the red vertices, player 2 the blue vertices, and player 3 the green vertices.}
    \label{fig:generaldeltaex}
\end{figure}

\begin{proof}{\textbf{Proof of Nonexistence for General $\Delta$.}}
We construct a partition exchange economy with exchange bound $\Delta$.
The compatibility graph $\graph$ consists of $60x$ vertex-disjoint type-A graphs and $180x$ vertex-disjoint type-B graphs, where $x\in\mathbb{N}$ is a parameter.

A type-A graph has $3(\Delta-1)$ vertices and is obtained from a triangle of bidirectional compatibility arcs by subdividing one direction of each pair.
The $i$-th type-A graph $\graph_i^A$ has vertices
\[
a_i^1,a_i^2,a_i^3,\qquad
a_i^{12}(1),\ldots,a_i^{12}(\Delta-2),\qquad
a_i^{13}(1),\ldots,a_i^{13}(\Delta-2),
\quad \text{and} \quad
a_i^{23}(1),\ldots,a_i^{23}(\Delta-2).
\]
Its compatibility arcs form the following three directed cycles:
\[
a_i^1\to a_i^{12}(1)\to\cdots\to a_i^{12}(\Delta-2)\to a_i^2\to a_i^1,
\]
\[
a_i^2\to a_i^{23}(1)\to\cdots\to a_i^{23}(\Delta-2)\to a_i^3\to a_i^2,
\]
\[
a_i^1\to a_i^{13}(1)\to\cdots\to a_i^{13}(\Delta-2)\to a_i^3\to a_i^1.
\]
If $i$ is even, player 1 owns
\(
a_i^1,a_i^{12}(1),\ldots,a_i^{12}(\Delta-2),
\)
player 2 owns
\(
a_i^2,a_i^{23}(1),\ldots,a_i^{23}(\Delta-2),
\)
and player 3 owns
\(
a_i^3,a_i^{13}(1),\ldots,a_i^{13}(\Delta-2).
\)

If $i$ is odd, player 1 owns
\(
a_i^1,a_i^{13}(1),\ldots,a_i^{13}(\Delta-2),
\)
player 2 owns
\(
a_i^2,a_i^{12}(1),\ldots,a_i^{12}(\Delta-2),
\)
and player 3 owns
\(
a_i^3,a_i^{23}(1),\ldots,a_i^{23}(\Delta-2).
\)

A type-B graph $\graph_i^B$ has $10\lfloor(\Delta-2)/2\rfloor+5$ vertices and is obtained from a 5-cycle of bidirectional compatibility arcs by subdividing each direction.
Let
\[
\delta:=\left\lfloor\frac{\Delta-2}{2}\right\rfloor.
\]
The $i$-th type-B graph has vertices
\[
\{b_i^l:l\in[5]\}
\quad \text{and} \quad
\left\{
b_i^{l(l+1)}(1),\ldots,b_i^{l(l+1)}(\delta),
b_i^{(l+1)l}(1),\ldots,b_i^{(l+1)l}(\delta)
:l\in[5]
\right\},
\]
where $5+1:=1$ in the upper indices; if $\delta=0$, no subdivision vertices are added.
For every $l\in[5]$, the compatibility arcs contain the directed cycle
\[
b_i^l
\to b_i^{l(l+1)}(1)
\to\cdots\to
b_i^{l(l+1)}(\delta)
\to b_i^{l+1}
\to b_i^{(l+1)l}(1)
\to\cdots\to
b_i^{(l+1)l}(\delta)
\to b_i^l.
\]

If $i\equiv 1\pmod 3$, player 1 owns
\[
b_i^1,b_i^{51}(1),\ldots,b_i^{51}(\delta),
b_i^{21}(1),\ldots,b_i^{21}(\delta),
b_i^2,b_i^{12}(1),\ldots,b_i^{12}(\delta),
b_i^{32}(1),\ldots,b_i^{32}(\delta),
\]
player 2 owns
\[
b_i^3,b_i^{23}(1),\ldots,b_i^{23}(\delta),
b_i^{43}(1),\ldots,b_i^{43}(\delta),
b_i^4,b_i^{34}(1),\ldots,b_i^{34}(\delta),
b_i^{54}(1),\ldots,b_i^{54}(\delta),
\]
and player 3 owns
\[
b_i^5,b_i^{45}(1),\ldots,b_i^{45}(\delta),
b_i^{15}(1),\ldots,b_i^{15}(\delta).
\]
If $i\equiv 2\pmod 3$, the same three ownership blocks are assigned, respectively, to players 2, 3, and 1. If $i\equiv 0\pmod 3$, they are assigned, respectively, to players 3, 1, and 2.

Thus, \(
|V|
=
60x\cdot 3(\Delta-1)
+
180x\cdot(10\delta+5)
=
O(x).
\)
Figure~\ref{fig:generaldeltaex} illustrates the two component types.

The feasible exchange cycles in a type-A graph $\graph_i^A$ are exactly the three $\Delta$-cycles described above. Any two intersect, so at most one can be selected.

The feasible exchange cycles in a type-B graph $\graph_i^B$ are exactly the five cycles
\[
\left\{
b_i^l,
b_i^{l(l+1)}(1),\ldots,b_i^{l(l+1)}(\delta),
b_i^{l+1},
b_i^{(l+1)l}(1),\ldots,b_i^{(l+1)l}(\delta)
\right\},
\qquad l\in[5].
\]
At most two such cycles can be selected simultaneously, and two can be selected if and only if they are nonadjacent in the underlying 5-cycle.

Hence, without additional altruistic donors, the three players together can cover at most
\[
60x\Delta+180x(4\delta+4)
=
\begin{cases}
840x(\delta+1), & \Delta\ \text{even},\\[2mm]
840x(\delta+1)+60x, & \Delta\ \text{odd}.
\end{cases}
\]

A single player can exchange only in the type-B components in which that player owns $4\delta+2$ vertices. There are $120x$ such components, and in each the player can cover $2\delta+2$ of its vertices. Thus every player can guarantee
\[
120x(2\delta+2)=240x(\delta+1)
\]
transplants by acting alone.

Now fix any two players. In every type-A component, their coalition has exactly one feasible cycle. In $30x$ of the components this cycle covers $\Delta-1$ vertices of the first player and one vertex of the second, and in the remaining $30x$ components the roles are reversed. Hence each coalition member receives
\[
30x\Delta
=
\begin{cases}
60x(\delta+1), & \Delta\ \text{even},\\[2mm]
60x(\delta+1)+30x, & \Delta\ \text{odd},
\end{cases}
\]
transplants from the type-A components.

Among the type-B components, there are $120x$ components in which one coalition member owns $4\delta+2$ vertices and the other owns $2\delta+1$ vertices. In each such component the coalition can cover $2\delta+2$ vertices, either splitting them as $(\delta+1,\delta+1)$ or assigning all $2\delta+2$ transplants to the member who owns $4\delta+2$ vertices. In the remaining $60x$ type-B components, both coalition members own $4\delta+2$ vertices and the unique jointly optimal choice gives each of them $2\delta+2$ transplants. Therefore the coalition obtains exactly
$480x(\delta+1)$
type-B transplants in total.

Moreover, every complementary split of these type-B transplants in multiples of $\delta+1$ between
\[
180x(\delta+1)
\quad\text{and}\quad
300x(\delta+1)
\]
is attainable. Indeed, start from the balanced choice in all $120x$ asymmetric components, which gives each player $240x(\delta+1)$ type-B transplants. For either player, there are exactly $60x$ asymmetric components in which that player is the large owner. Switching any one of those components from the balanced choice to the concentrated choice transfers exactly $\delta+1$ transplants from the other player to that player. Hence every split
\(
\bigl((240x+q)(\delta+1),(240x-q)(\delta+1)\bigr),
\quad
q\in\{-60x,\ldots,60x\},
\)
is feasible.

Suppose, toward a contradiction, that the $3x$-supplemented core is nonempty. Consider any exchange in it, and let the $3x$ additional altruistic donors be the centrally supplied donors. Delete these donors from every exchange cycle containing at least one of them. The remaining original vertices on such a cycle form directed path segments in the original compatibility graph. Since the original graph is a disjoint union of type-A and type-B components and every directed path in either component contains fewer than $10(\delta+1)$ original vertices, each supplied donor can account for fewer than $10(\delta+1)$ additional covered original vertices. Consequently, the $3x$ supplied donors can increase total utility by at most
\(
30x(\delta+1).
\)
Thus the total utility in the supplemented exchange is at most
\[
\begin{cases}
870x(\delta+1), & \Delta\ \text{even},\\[2mm]
870x(\delta+1)+60x, & \Delta\ \text{odd}.
\end{cases}
\]
The two worst-off players therefore receive, together, at most
\[
\begin{cases}
580x(\delta+1), & \Delta\ \text{even},\\[2mm]
580x(\delta+1)+40x, & \Delta\ \text{odd}.
\end{cases}
\]

Let these two players have utilities $u_1$ and $u_2$. Since neither can block alone,
\[
u_i\ge 240x(\delta+1),\qquad i=1,2.
\]
Combining this lower bound with the preceding bound on $u_1+u_2$ gives
\[
u_i
\le
\begin{cases}
340x(\delta+1), & \Delta\ \text{even},\\[2mm]
340x(\delta+1)+40x, & \Delta\ \text{odd},
\end{cases}
\qquad i=1,2.
\]
If the two players deviate together, each receives from the type-A components
\[
A:=
\begin{cases}
60x(\delta+1), & \Delta\ \text{even},\\[2mm]
60x(\delta+1)+30x, & \Delta\ \text{odd}.
\end{cases}
\]
Hence, to strictly improve on $u_i$, player $i$ needs from the type-B components strictly more than $u_i-A$. The preceding bound implies
\[
u_i-A\le 290x(\delta+1),\qquad i=1,2.
\]
Furthermore,
\[
(u_1-A)+(u_2-A)
\le
\begin{cases}
460x(\delta+1), & \Delta\ \text{even},\\[2mm]
460x(\delta+1)-20x, & \Delta\ \text{odd},
\end{cases}
\]
so the coalition has at least $20x(\delta+1)$ of aggregate slack relative to its total type-B capacity $480x(\delta+1)$.

To make the final split explicit, write
\[
r_i:=\frac{u_i-A}{\delta+1},\qquad i=1,2.
\]
An attainable type-B split has the form
\[
\bigl((240x+q)(\delta+1),(240x-q)(\delta+1)\bigr),
\qquad q\in\{-60x,\ldots,60x\}.
\]
We need an integer $q$ satisfying
\[
q>r_1-240x
\qquad\text{and}\qquad
q<240x-r_2.
\]
The length of this interval is
\(
480x-(r_1+r_2)\ge 20x,
\)
by the preceding bound on $(u_1-A)+(u_2-A)$. Moreover, $r_i\le290x$, so its lower endpoint is at most $50x$ and its upper endpoint is at least $-50x$. Hence the interval contains an integer $q\in[-60x,60x]$. For this attainable split, each coalition member receives strictly more than $u_i-A$ from the type-B components and therefore strictly more than $u_i$ in total. Thus the two players form a blocking coalition, contradicting supplemented-core stability.

Therefore the $3x$-supplemented core is empty. Since $|V|=O(x)$, this proves the claimed linear lower bound.\qed{}
\end{proof}

\section{Omitted Proofs for Pairwise Exchange}

\subsection{Proof of Remark~\ref{rem:localchange}}
\begin{proof}{Proof}
Let $\mathcal{C}$ be the independent collection of vertex-disjoint odd
cycles obtained in the proof of Theorem~\ref{thm:arbitrary}. 
Denote by $V(c)$ the vertex set of a cycle $c$. 
For a collection of cycles $\mathcal{D}$, write
\(
V(\mathcal{D}) := \bigcup_{c\in\mathcal{D}} V(c),
\) 
and let
\(
U:=\bigcup_{i\in\N}U^i.
\)

Partition $\mathcal{C}$ into
\[
\mathcal{C}_A
:=
\{c\in\mathcal{C}:V(c)\setminus U\neq\varnothing\}\;\; 
\text{ and } \;\;
\mathcal{C}_U
:=
\{c\in\mathcal{C}:V(c)\subseteq U\}.
\]
Thus, every cycle in $\mathcal{C}_A$ contains an existing altruistic
vertex, whereas every vertex of every cycle in $\mathcal{C}_U$ is
utility-bearing.

For each $c\in\mathcal{C}_A$, choose an altruistic vertex
$a_c\in V(c)\setminus U$ and leave $a_c$ uncovered. Since $c$ is an odd
cycle, deleting $a_c$ leaves a path with an even number of vertices,
which admits a perfect matching. Hence, every utility-bearing vertex
of $c$ can be covered without using an additional altruistic donor.

It remains to choose the uncovered vertex from each cycle in
$\mathcal{C}_U$. Construct a bipartite graph $\mathcal{H}$ whose left
vertex set is $\mathcal{C}_U$ and whose right vertex set is $\N$.
A cycle $c\in\mathcal{C}_U$ is adjacent to organization $i$ if
\[
V(c)\cap U^i\neq\varnothing.
\]
Give organization $i$ capacity
\[
q_i
:=
\left\lceil
\frac{|U^i\cap V(\mathcal{C}_U)|}{3}
\right\rceil.
\]

We claim that every cycle in $\mathcal{C}_U$ can be assigned to an
adjacent organization while assigning at most $q_i$ cycles to
organization $i$. By the capacitated version of Hall's theorem, it
suffices to show that, for every
$\mathcal{S}\subseteq\mathcal{C}_U$,
\[
|\mathcal{S}|
\leq
\sum_{i\in N_{\mathcal{H}}(\mathcal{S})}q_i,
\]
where $N_{\mathcal{H}}(\mathcal{S})$ denotes the organizations adjacent
to at least one cycle in $\mathcal{S}$.

The cycles are vertex-disjoint, every cycle in $\mathcal{C}_U$ has at
least three vertices, and all of these vertices are utility-bearing.
Consequently,
\[
\sum_{i\in N_{\mathcal{H}}(\mathcal{S})}
   |U^i\cap V(\mathcal{C}_U)|
\geq
\sum_{c\in\mathcal{S}}|V(c)|
\geq
3|\mathcal{S}|.
\]
It follows that
\[
\sum_{i\in N_{\mathcal{H}}(\mathcal{S})}q_i
=
\sum_{i\in N_{\mathcal{H}}(\mathcal{S})}
\left\lceil
\frac{|U^i\cap V(\mathcal{C}_U)|}{3}
\right\rceil
\geq
\frac{1}{3}
\sum_{i\in N_{\mathcal{H}}(\mathcal{S})}
|U^i\cap V(\mathcal{C}_U)|
\geq
|\mathcal{S}|.
\]
The capacitated Hall condition therefore holds.

Let
\(
\phi:\mathcal{C}_U\longrightarrow\N
\)
be such an assignment. For every $c\in\mathcal{C}_U$, choose a vertex
\(
v_c\in V(c)\cap U^{\phi(c)}
\)
and leave $v_c$ uncovered. Again, deleting $v_c$ from the odd cycle
$c$ leaves an even path and hence a perfect matching of the remaining
vertices.

Combining these matchings over all cycles in $\mathcal{C}$ with the
integral edges fixed in the proof of
Theorem~\ref{thm:arbitrary} gives a matching in which the only
uncovered utility-bearing vertices from the residual odd cycles are
the vertices $v_c$, $c\in\mathcal{C}_U$. Add one suitably compatible
altruistic donor for each such vertex.

The number of additional donors used for patients of organization $i$
is
\[
|\phi^{-1}(i)|
\leq q_i
=
\left\lceil
\frac{|U^i\cap V(\mathcal{C}_U)|}{3}
\right\rceil
\leq
\left\lceil
\frac{|U^i|}{3}
\right\rceil.
\]
Moreover, the total number of additional donors is
\[
|\mathcal{C}_U|
\leq
|\mathcal{C}|
\leq
\nu(\graph).
\]

After supplementation, every utility-bearing vertex in the residual
odd cycles is covered. Thus, the resulting exchange satisfies the
same organization-level utility lower bounds as in the proof of
Theorem~\ref{thm:arbitrary}. Lemma~\ref{thm:ApplyScarf} therefore
implies that it belongs to the $\nu(\graph)$-supplemented core. 

\end{proof}

\subsection{Proof of Lemma~\ref{thm:logbound}}

\begin{proof}{Proof} 
Let $n:=|V|$ and condition on an arbitrary assignment of the vertices
to the $g$ groups in $\Phi$. Recall that
\[
p:=\min\{p_{ab}:p_{ab}>0\}>0.
\]
No assumption is imposed on the diagonal probabilities $p_{aa}$.

For each vertex $v$, define its augmented class as
\[
\tau(v)=(a,r)\in \Phi\times\{\mathrm{u},\mathrm{a}\},
\]
where $a$ is the group of $v$, and $r=\mathrm{u}$ if $v$ is
utility-bearing and $r=\mathrm{a}$ if $v$ represents an existing
altruistic donor. Thus, there are at most $2g$ augmented classes.

For two augmented classes $\theta$ and $\theta'$, let
$q_{\theta\theta'}$ denote the probability that two distinct vertices
in these classes are mutually compatible. Each of the two required
directed arcs is either deterministic, because its target is a dummy
patient, or appears independently with probability at least $p$.
Consequently,
\[
q_{\theta\theta'}>0
\quad\Longrightarrow\quad
q_{\theta\theta'}\ge p^2.
\]
In particular, for two utility-bearing classes corresponding to groups
$a$ and $b$,
\[
q_{\theta\theta'}=p_{ab}p_{ba},
\]
including $q_{\theta\theta}=p_{aa}^2$ when $a=b$. Moreover, mutual-edge
indicators associated with distinct unordered pairs of vertices are
independent.

Let $\indep(\graph)$ denote the maximum number of independent edges in
the mutual-compatibility graph; equivalently, $\indep(\graph)$ is the
maximum size of an induced matching. Let
\[
X(\graph):=\indep(\graph).
\]
Selecting one edge from each cycle in an independent collection of odd
cycles gives an induced matching. Therefore,
\[
\nu(\graph)\le X(\graph).
\]
It is thus sufficient to show that
\(
\mathbb{E}[X]=\mathcal{O}(\log n)
\)
and that $X=\mathcal{O}(\log n)$ with probability approaching one.

Let
\[
h:=\binom{2g+1}{2},
\]
which bounds the number of unordered pairs of augmented classes.
Suppose that $\graph$ contains an induced matching of size $k$. By the
pigeonhole principle, this matching contains at least
\(
\ell:=\left\lceil\frac{k}{h}\right\rceil
\)
edges whose endpoints have the same unordered pair of augmented
classes, say $\{\theta,\theta'\}$.

Consider a fixed collection of $\ell$ disjoint candidate edges of this
class-pair type. If $q_{\theta\theta'}=0$, the probability that these
candidate edges form an induced matching is zero. Suppose therefore
that $q_{\theta\theta'}>0$.

First, let $\theta\neq\theta'$, and write the candidate edges as
$\{x_i,y_i\}$, where $x_i$ belongs to $\theta$ and $y_i$ belongs to
$\theta'$. For these edges to form an induced matching, every
off-matching pair $\{x_i,y_j\}$ with $i\neq j$ must be a nonedge.
There are $\ell(\ell-1)$ such pairs.

If $\theta=\theta'$, all $2\ell$ endpoints belong to the same augmented
class. In this case, inducedness requires all nonmatching pairs among
the endpoints to be absent. In particular, it requires at least
$\ell(\ell-1)$ distinct potential edges to be absent.

Thus, in either case,
\[
\Pr(\text{the fixed candidate edges form an induced matching})
 \le (1-p^2)^{\ell(\ell-1)}
 \le \exp\{-p^2\ell(\ell-1)\}.
\]
There are at most $n^{2\ell}$ collections of $\ell$ candidate edges
and at most $h$ possible unordered class-pair types. Hence,
\[
\Pr(X\ge h\ell)
 \le h n^{2\ell}\exp\{-p^2\ell(\ell-1)\}.
\]

Set
\[
\ell_0:=
\left\lceil\frac{8}{p^2}\log(n+2)\right\rceil
\qquad\text{and}\qquad
k_0:=h\ell_0.
\]
For all sufficiently large $n$,
\[
p^2\ell_0(\ell_0-1)
 \ge 4\ell_0\log(n+2).
\]
It follows that for large enough $n$,
\[
\Pr(X\ge k_0)
 \le h\exp\{-2\ell_0\log(n+2)\}
 \le n^{-3}.
\]
Therefore,
\[
\Pr\bigl(\nu(\graph)\le k_0\bigr)
 \ge \Pr(X\le k_0)
 \longrightarrow 1.
\]
Since $X\le n$, we also have
\[
\begin{aligned}
\mathbb{E}[\nu(\graph)]
 &\le \mathbb{E}[X] \\
 &\le k_0+n\Pr(X\ge k_0) \\
 &\le h\left\lceil\frac{8}{p^2}\log(n+2)\right\rceil+n^{-2}
  =\mathcal{O}_{g,p}(\log n).
\end{aligned}
\]
The bounds are uniform over the assignment of vertices to groups, so
they continue to hold without conditioning on that assignment.
\end{proof}

\medskip
\section{Logarithmic bound beyond pairwise exchange}

\subsection{Three-Way Exchanges in the Random Compatibility Model}
\label{app:three-way-rounding}

This appendix proves
Lemma~\ref{lem:logarithmic-three-rounding}. The proof first aggregates
a fractional exchange by finitely many augmented cycle patterns, then
rounds the aggregate pattern counts, and finally realizes the rounded
counts in the random compatibility graph.

\subsubsection*{Augmented classes and cycle patterns}

Let
\[
V=\bigcup_{i\in\N}V^i,
\qquad
U=\bigcup_{i\in\N}U^i,
\qquad
A:=V\setminus U,
\]
where $A$ is the set of existing altruistic donors, represented by
vertices with dummy patients. Let $m:=|V|$.

Recall that the set of random-graph groups is denoted by $\Phi$, where
$|\Phi|=g$. Every vertex belongs to one group. If the patient
represented by $v$ is not a dummy patient, then, for vertices $u$ and
$v$ in groups $a$ and $b$, respectively, the arc $(u,v)$ appears
independently with probability $p_{ab}$. If $v\in A$, then its patient
is a dummy patient and $(u,v)$ is present deterministically for every
$u\neq v$.

We impose no restriction on the diagonal probabilities $p_{aa}$ and define
\(
p:=\min\{p_{ab}:p_{ab}>0\}.
\)
Both $g$ and $p>0$ are constants. For a vertex $v$, let
\[
\tau(v)=
\begin{cases}
\mathrm{u},&v\in U,\\
\mathrm{a},&v\in A.
\end{cases}
\]
An \emph{augmented class} is a pair
\[
    \theta=(a,r)\in\Phi\times\{\mathrm{u},\mathrm{a}\},
\]
where $a$ is the random-graph group and $r$ records whether the vertex
is utility-bearing or altruistic. Let $V_\theta$ denote the vertices
of augmented class $\theta$.

For augmented classes $\theta=(a,r)$ and $\theta'=(b,r')$, we say that
the transition $\theta\to\theta'$ is \emph{admissible} if either
$r'=\mathrm{a}$, in which case the target patient is dummy and the arc
is deterministic, or $r'=\mathrm{u}$ and $p_{ab}>0$.

An \emph{augmented cycle pattern} of length $k\in\{2,3\}$ is a cyclic
equivalence class of ordered tuples
\[
    \pi=(\theta_1,\ldots,\theta_k)
\]
such that every transition $\theta_j\to\theta_{j+1}$ is admissible,
where indices are taken modulo $k$. Tuples that differ only by cyclic
rotation represent the same pattern. For each pattern, fix one
representative. Let $\mathcal P$ denote the set of augmented cycle
patterns, and let $P:=|\mathcal P|$. Since $g$ is constant,
\[
    P\le (2g)^2+(2g)^3=\mathcal O(g^3).
\]
For $\pi\in\mathcal P$, let $|\pi|\in\{2,3\}$ denote its length. For an
augmented class $\theta$, let $\mu_\theta^\pi$ denote the number of
occurrences of $\theta$ in $\pi$. Define
\[
    w(\pi):=
    \sum_{\theta=(a,\mathrm{u})}\mu_\theta^\pi,
\]
the number of utility-bearing slots of $\pi$. 
A pattern may contain the same augmented class more than once; this
can occur either because same-group compatibility has positive
probability or because a transition enters an altruistic vertex and
is therefore deterministic.


\subsubsection*{A uniform random-graph event}

\begin{definition}\label{def:robust-cycle-event}
For an integer $h\ge1$, let $\mathcal R_h$ be the event that the
following properties hold.
\begin{enumerate}
\item[(R1)]
For every pair of augmented classes $\theta,\theta'$ for which the
transition $\theta\to\theta'$ is admissible and random, and for all
disjoint sets
\(
    X\subseteq V_\theta,
    \qquad
    Y\subseteq V_{\theta'}
\)
with $|X|,|Y|\ge h$, there is an arc from $X$ to $Y$.

\item[(R2)]
For every augmented length-two pattern represented by
$(\theta_1,\theta_2)$ and all disjoint sets
$X\subseteq V_{\theta_1}$ and $ Y\subseteq V_{\theta_2}$
with $|X|,|Y|\ge h$, there are $x\in X$ and $y\in Y$ such that
\(
    x\to y\to x
\)
is a directed two-cycle.

\item[(R3)]
For every augmented length-three pattern represented by
$(\theta_1,\theta_2,\theta_3)$, every set
\(
    M\subseteq V_{\theta_1}\times V_{\theta_2}
\)
of pairwise vertex-disjoint ordered pairs satisfying $x\to y$ for
every $(x,y)\in M$, and every set
\(
    Z\subseteq V_{\theta_3}
\)
disjoint from the vertices occurring in $M$, if $|M|,|Z|\ge h$, then
there are $(x,y)\in M$ and $z\in Z$ such that
\(
    y\to z\to x.
\)
\end{enumerate}
\end{definition}
The disjointness conditions above ensure that the relevant directed-edge
variables are distinct even when two or more augmented classes coincide.

Transitions whose target is altruistic are deterministic and
therefore cannot cause any of these properties to fail. Every random
transition in an admissible pattern has probability at least $p$.

\begin{lemma}\label{lem:robust-cycle-event}
There is a constant $K=K(g,p)$ such that, for
\(
    h:=\left\lceil K\log(m+2)\right\rceil,
\)
we have
\[
    \Pr(\mathcal R_h)\ge1-(m+2)^{-4}.
\]
\end{lemma}

\begin{proof}{Proof.}
Write $L:=\log(m+2)$ and choose $K\ge18/p^2$. We establish each
property for fixed augmented classes or a fixed pattern and then take
a union bound over the constantly many choices.

First consider (R1). If the transition is deterministic, the assertion
holds automatically. Otherwise, fix sets $X,Y$ of sizes $s,t$,
respectively. The probability that there is no arc from $X$ to $Y$ is
at most
\[
    (1-p)^{st}
    \le \exp(-pst)
    \le \exp(-p^2st).
\]
There are at most $m^{s+t}$ choices for $X$ and $Y$. Since
$s,t\ge h$,
\[
    st\ge\frac{h}{2}(s+t).
\]
Consequently,
\begin{align*}
    (s+t)\log m-p^2st
    &\le (s+t)L-\frac{p^2h}{2}(s+t)\\
    &\le -8(s+t)L.
\end{align*}
Thus, the probability that (R1) fails for fixed augmented classes is
at most
\[
\begin{aligned}
    \sum_{s,t\ge h}
    \exp\bigl((s+t)\log m-p^2st\bigr)
    &\le
    \sum_{s,t\ge h}(m+2)^{-8(s+t)}\\
    &=\mathcal O\bigl((m+2)^{-16}\bigr).
\end{aligned}
\]

For (R2), fix sets $X,Y$ of sizes $s,t$. For any $x\in X$ and $y\in
Y$, each required arc is either deterministic or appears
independently with probability at least $p$. Hence, the probability
that $x\to y\to x$ is a directed two-cycle is either one or at least
$p^2$. The events corresponding to distinct pairs $(x,y)$ depend on
disjoint directed-edge variables. Therefore,
\[
\begin{aligned}
    &\Pr(\text{there is no directed two-cycle between }X\text{ and }Y)\\
    &\hspace{3em}\le
    (1-p^2)^{st}
    \le \exp(-p^2st).
\end{aligned}
\]
The same calculation as for (R1) gives a failure probability
$\mathcal O((m+2)^{-16})$ for each fixed length-two pattern.

Finally, consider (R3). Fix a family $M$ of $s$ pairwise
vertex-disjoint ordered pairs and a set $Z$ of size $t$. For a fixed
$(x,y)\in M$ and $z\in Z$, each of the arcs $y\to z$ and $z\to x$ is
either deterministic or occurs independently with probability at
least $p$. The probability that both arcs exist is therefore either
one or at least $p^2$. As $(x,y)$ ranges over $M$ and $z$ ranges over
$Z$, all non-deterministic directed-edge variables used by these
completion events are distinct. Hence,
\[
    \Pr(\text{no pair in }M\text{ can be completed using }Z)
    \le \exp(-p^2st).
\]
There are at most $m^{2s}$ choices for $M$ and at most $m^t$ choices
for $Z$. Since $s,t\ge h$,
\[
    st\ge\frac{h}{3}(2s+t).
\]
Therefore,
\begin{align*}
    (2s+t)\log m-p^2st
    &\le (2s+t)L-\frac{p^2h}{3}(2s+t)\\
    &\le -5(2s+t)L.
\end{align*}
The failure probability for a fixed length-three pattern is at most
\[
    \sum_{s,t\ge h}(m+2)^{-5(2s+t)}
    =
    \mathcal O\bigl((m+2)^{-15}\bigr).
\]
There are only $\mathcal O(g^3)$ augmented patterns. A union bound
proves the result for all sufficiently large $m$; increasing $K$ if
necessary handles the remaining values of $m$.
\qed{}
\end{proof}

\subsubsection*{Realizing prescribed pattern counts}

We use the following consequence of Hall's theorem.

\begin{lemma}\label{lem:empty-rectangle-matching}
Let $H=(L,R;F)$ be a bipartite graph with $|L|=|R|=q$. Suppose that
there are no sets $X\subseteq L$ and $Y\subseteq R$ such that
$|X|,|Y|\ge h$ and there is no edge between $X$ and $Y$. Then $H$ has
a matching of size at least $q-h$.
\end{lemma}

\begin{proof}{Proof.}
Let a maximum matching have size $q-d$. By the deficiency form of
Hall's theorem, there is a set $X\subseteq L$ such that
\[
    d=|X|-|N_H(X)|.
\]
Set $Y:=R\setminus N_H(X)$. There is no edge between $X$ and $Y$, and
\[
    d=|X|+|Y|-q.
\]
If $|X|<h$, then $d\le |X|<h$. Similarly, if $|Y|<h$, then
$d\le |Y|<h$. Otherwise, $X$ and $Y$ contradict the assumed absence of
a large empty rectangle. Hence $d<h$, proving the claim.
\qed{}
\end{proof}

\begin{lemma}\label{lem:realize-cycle-pattern}
Assume that $\mathcal R_h$ holds. Let
\[
    \pi=(\theta_1,\ldots,\theta_k)\in\mathcal P,
    \qquad k\in\{2,3\},
\]
and suppose that, for every occurrence of a slot in $\pi$, we are given
a set of exactly $q$ vertices of the corresponding augmented class.
Assume that the sets assigned to distinct slots are pairwise disjoint,
even when two slots have the same augmented class.

Then $\graph$ contains vertex-disjoint directed cycles of pattern
$\pi$ that leave at most $2h$ supplied vertices uncovered in each
slot. Consequently, at most
\[
    2h|\pi|\le6h
\]
supplied vertices are left uncovered in total.
\end{lemma}

\begin{proof}{Proof.}
If $q<2h$, discard all supplied vertices. The claimed bound follows
immediately. Suppose therefore that $q\ge2h$.

First let $k=2$, and let the two supplied slot sets be $X_1$ and
$X_2$. Construct a bipartite graph between $X_1$ and $X_2$, where
$x_1\in X_1$ and $x_2\in X_2$ are adjacent if
\[
    x_1\to x_2\to x_1
\]
is a directed two-cycle. By (R2), this bipartite graph has no empty
rectangle with both sides of size at least $h$.
Lemma~\ref{lem:empty-rectangle-matching} gives a matching of size at
least $q-h$. Thus, at most $h$ vertices remain uncovered in each slot.

Now let $k=3$, and denote the supplied slot sets by $X_1,X_2,X_3$.
Consider first the transition from the first slot to the second. If
this transition is deterministic, choose an arbitrary perfect matching
from $X_1$ to $X_2$. Otherwise, by (R1) and
Lemma~\ref{lem:empty-rectangle-matching}, the bipartite graph of arcs
from $X_1$ to $X_2$ has a matching $M$ of size at least $q-h$.

Thus, in either case, we obtain a family $M$ of at least $q-h$
pairwise vertex-disjoint directed paths $x_1\to x_2$. Choose an
arbitrary set $X_3'\subseteq X_3$ with $|X_3'|=|M|$. Construct an
auxiliary bipartite graph with left side $M$ and right side $X_3'$. A
path $(x_1,x_2)\in M$ is adjacent to $x_3\in X_3'$ if
\[
    x_2\to x_3\to x_1.
\]
By (R3), this auxiliary graph has no empty rectangle with both sides
of size at least $h$. Hence,
Lemma~\ref{lem:empty-rectangle-matching} gives a matching of size at
least
\[
    |M|-h\ge q-2h.
\]
Every matched pair gives a directed cycle
\[
    x_1\to x_2\to x_3\to x_1.
\]
Therefore, at most $2h$ supplied vertices remain uncovered in each
slot.
\qed{}
\end{proof}

\subsubsection*{Aggregate-pattern rounding}

\begin{lemma}\label{lem:random-three-rounding}
Assume that $\mathcal R_h$ holds. Let $y$ be a fractional exchange
using cycles of size at most three, and define
\[
    b_i:=
    \left\lfloor
        \sum_{c\in\mathcal C_3}\gamma_c^iy_c
    \right\rfloor
    \qquad(i\in\N).
\]
There is an integral exchange $\Ex$, using cycles of size at most three
and only vertices from the original pool $V$, such that
\[
    \sum_{i\in\N}
    \bigl(b_i-u_i(\Ex)\bigr)^+
    \le9Ph.
\]
\end{lemma}

\begin{proof}{Proof.}
For each augmented pattern $\pi\in\mathcal P$, let
\[
    \lambda_\pi:=
    \sum_{\substack{c\in\mathcal C_3:\\
                    c\text{ has augmented pattern }\pi}}
    y_c,
    \qquad
    q_\pi:=\lfloor\lambda_\pi\rfloor.
\]
For every augmented class $\theta$, define its fractional and rounded
usages by
\[
    R_\theta:=
    \sum_{\pi\in\mathcal P}\mu_\theta^\pi\lambda_\pi,
    \qquad
    S_\theta:=
    \sum_{\pi\in\mathcal P}\mu_\theta^\pi q_\pi.
\]
Since $q_\pi\le\lambda_\pi$, we have $S_\theta\le R_\theta$.
Moreover, the vertex-capacity constraints of the fractional exchange
imply
\[
    R_\theta\le|V_\theta|.
\]

Only utility-bearing slots matter for organizational utilities. Define
the loss caused by flooring the augmented-pattern weights on
utility-bearing slots as
\[
\begin{aligned}
    \delta_U
    &:=
    \sum_{\theta=(a,\mathrm{u})}(R_\theta-S_\theta)\\
    &=
    \sum_{\pi\in\mathcal P}
    w(\pi)(\lambda_\pi-q_\pi).
\end{aligned}
\]
Since $w(\pi)\le|\pi|\le3$,
\begin{equation}\label{eq:utility-pattern-floor-loss}
    \delta_U<3P.
\end{equation}

\paragraph{Step 1: Assigning the rounded slots to organizations.}
For every organization $i$ and every utility-bearing augmented class
$\theta=(a,\mathrm{u})$, let
\[
    U^{i,\theta}:=U^i\cap V_\theta,
    \qquad
    n_{i,\theta}:=|U^{i,\theta}|,
\]
and define
\[
    R_{i,\theta}:=
    \sum_{c\in\mathcal C_3}
    |c\cap U^{i,\theta}|y_c.
\]
Then
\[
    0\le R_{i,\theta}\le n_{i,\theta},
    \qquad
    \sum_{i\in\N}R_{i,\theta}=R_\theta,
\]
and
\[
    \sum_{\theta=(a,\mathrm{u})}R_{i,\theta}
    =
    \sum_{c\in\mathcal C_3}\gamma_c^iy_c
    \ge b_i.
\]

For each utility-bearing class $\theta$, reduce the vector
$(R_{i,\theta})_{i\in\N}$ by total amount $R_\theta-S_\theta$,
obtaining numbers $F_{i,\theta}$ such that
\[
    0\le F_{i,\theta}\le R_{i,\theta},
    \qquad
    \sum_{i\in\N}F_{i,\theta}=S_\theta.
\]
The total amount removed over all utility-bearing classes is
$\delta_U$. Therefore,
\begin{equation}\label{eq:fractional-utility-deficit}
    \sum_{i\in\N}
    \left(
        b_i-
        \sum_{\theta=(a,\mathrm{u})}F_{i,\theta}
    \right)^+
    \le\delta_U.
\end{equation}

We now round this fractional assignment without increasing its total
utility deficit.  Consider the following maximum-cost flow network.
There is a source $s$, one node for every utility-bearing class
$\theta$, one node for every organization $i$, and a sink $t$.  The
arc $s\to\theta$ has lower and upper capacity both equal to the
integer $S_\theta$, while the arc $\theta\to i$ has capacity
$n_{i,\theta}$.  Finally, from each organization node $i$ to $t$ add
two parallel arcs: a unit-profit arc of capacity $b_i$, and a
zero-profit arc of capacity
$B:=\sum_{\theta=(a,\mathrm{u})}S_\theta$ (which is large enough to be
nonbinding).  Thus a unit sent on $\theta\to i$ represents assigning
one selected vertex of class $\theta$ to organization $i$, and the
capacity $n_{i,\theta}$ enforces that no more such vertices are
assigned than actually exist.

All supplies and arc capacities in this network are integral.  Hence
the feasible-flow polytope is integral: after the usual elimination of
fixed lower bounds, its constraint matrix is the node--arc incidence
matrix of a directed graph together with bound constraints, and is
therefore totally unimodular.  In particular, the maximum-cost flow
problem has an integral optimum.

To see that the fractional assignment $F$ gives a feasible fractional
flow, put $F_{i,\theta}$ units on the arc $\theta\to i$.  For
$x_i:=\sum_{\theta=(a,\mathrm{u})}F_{i,\theta}$, send
$\min\{b_i,x_i\}$ units from $i$ to $t$ on the unit-profit arc and the
remaining $x_i-\min\{b_i,x_i\}$ units on the zero-profit arc.  Its
profit is therefore
\[
    \sum_{i\in\N}\min\{b_i,x_i\}
    =
    \sum_{i\in\N}b_i-
    \sum_{i\in\N}(b_i-x_i)^+
    \ge
    \sum_{i\in\N}b_i-\delta_U,
\]
where the inequality follows from
\eqref{eq:fractional-utility-deficit}.  Let an integral optimum of the
flow problem be fixed, and let $m_{i,\theta}$ be its flow on
$\theta\to i$.  Flow conservation and the arc capacities give
\[
    0\le m_{i,\theta}\le n_{i,\theta},
    \qquad
    \sum_{i\in\N}m_{i,\theta}=S_\theta
\]
for every utility-bearing class $\theta$.  Moreover, at optimality the
unit-profit arc out of organization $i$ carries
$\min\{b_i,\sum_\theta m_{i,\theta}\}$ units; otherwise flow could be
moved from the zero-profit arc to the unit-profit arc and strictly
increase the objective.  Since the integral optimum has profit at
least that of the fractional flow above, we obtain
\begin{equation}\label{eq:integer-utility-deficit}
    \sum_{i\in\N}
    \left(
        b_i-
        \sum_{\theta=(a,\mathrm{u})}m_{i,\theta}
    \right)^+
    \le
    \lceil\delta_U\rceil
    \le3P.
\end{equation}

\paragraph{Step 2: Selecting vertices for the rounded patterns.}
For each organization $i$ and each utility-bearing class $\theta$,
choose exactly $m_{i,\theta}$ vertices from $U^{i,\theta}$. Thus,
exactly $S_\theta$ selected utility-bearing vertices belong to class
$\theta$.

For every altruistic class $\theta=(a,\mathrm{a})$,
\[
    S_\theta\le R_\theta\le|V_\theta|.
\]
Choose an arbitrary set of exactly $S_\theta$ existing altruistic
vertices from $V_\theta$. These vertices are needed as cycle
positions, but their coverage contributes no utility.

For every augmented class $\theta$,
\[
    S_\theta=
    \sum_{\pi\in\mathcal P}\mu_\theta^\pi q_\pi.
\]
We can therefore partition the selected vertices of class $\theta$
among all occurrences of $\theta$ in the pattern copies so that every
pattern $\pi$ receives $q_\pi$ pairwise disjoint vertex sets for each
of its slots, each of size $q_\pi$. This remains valid when the same
augmented class occurs more than once in a pattern because distinct
slot occurrences receive disjoint sets.

\paragraph{Step 3: Realizing the rounded patterns.}
Apply Lemma~\ref{lem:realize-cycle-pattern} separately to every
pattern $\pi$. At most
\[
    2h|\pi|\le6h
\]
selected vertices assigned to one pattern remain uncovered. Hence,
this gives an integral exchange $\Ex$ in which at most $6Ph$ selected
vertices remain uncovered in total.

Only uncovered utility-bearing vertices can reduce organizational
utility. Consequently, the total utility deficit created in the
realization step is at most $6Ph$. Combining this with
\eqref{eq:integer-utility-deficit} gives
\[
    \sum_{i\in\N}
    \bigl(b_i-u_i(\Ex)\bigr)^+
    \le3P+6Ph
    \le9Ph,
\]
where the last inequality uses $h\ge1$.
\qed{}
\end{proof}

\begin{proof}{Proof of
Lemma~\ref{lem:logarithmic-three-rounding}.}
Let
\[
    h:=\left\lceil K\log(m+2)\right\rceil,
\]
where $K=K(g,p)$ is chosen as in
Lemma~\ref{lem:robust-cycle-event}. By that lemma,
\[
    \Pr(\mathcal R_h)\ge1-(m+2)^{-4}.
\]
Conditional on $\mathcal R_h$,
Lemma~\ref{lem:random-three-rounding} applies to every fractional
exchange $y$ and gives an integral exchange $\Ex$ satisfying
\[
    \sum_{i\in\N}
    \left(
        \left\lfloor
            \sum_{c\in\mathcal C_3}\gamma_c^iy_c
        \right\rfloor
        -
        u_i(\Ex)
    \right)^+
    \le9Ph.
\]
Because
\[
    P\le(2g)^2+(2g)^3
\]
and $h=\mathcal O_{g,p}(\log(m+2))$, there is a constant
$C_0=C_0(g,p)$ such that
\[
    9Ph\le C_0\log(m+2).
\]
The event $\mathcal R_h$ is defined entirely in terms of the realized
compatibility graph, so the conclusion holds simultaneously for every
fractional exchange $y$.
\qed{}
\end{proof}

\subsection{Extending the logarithmic bound to constant $\Delta$}\label{app:fixed-cycle-length}  
In this section, we show that the argument of Theorem~\ref{thm:cycle-random} extends to every fixed constant $\Delta$.

\begin{theorem}
Consider a partition exchange economy with a constant exchange bound $\Delta$, where the compatibility graph on the original vertices $V$, including the existing altruistic donors in $V\setminus U$, is generated according to the random model above.

Then there is a constant $C=C(g,p,\Delta)$ such that, if $d(\graph)$ denotes the minimum number of \emph{additional} altruistic donors required for a supplemented-core exchange, then
\[
\mathbb{E}[d(\graph)]
\le
C\log(|V|+2)
\]
and
\[
\Pr\!\left(
d(\graph)\le C\log(|V|+2)
\right)
\ge
1-(|V|+2)^{-4}.
\]
The expectation and probability are taken over the realization of the random compatibility arcs. The supplemented-core exchange can be chosen Pareto-optimal.
\end{theorem}

\begin{proof}{Proof}
Let $m:=|V|$, and let $y^*$ be the fractional exchange supplied by Lemma~\ref{thm:ApplyScarf}. It is therefore enough to construct, with the claimed number of additional donors, an integral exchange whose utility for every organization $i$ is at least the corresponding floor induced by $y^*$.

For $2\le k\le\Delta$, define an augmented pattern of length $k$ as a cyclic tuple
\(
\pi=(\theta_1,\ldots,\theta_k)
\)
of augmented classes with admissible consecutive transitions. Let $\mathcal{P}_\Delta$ be the set of all such patterns and let
\(
P_\Delta:=|\mathcal{P}_\Delta|.
\)
Since $g$ and $\Delta$ are fixed,
\[
P_\Delta
\le
\sum_{k=2}^{\Delta}(2g)^k
=
O_{g,\Delta}(1).
\]

We use the same aggregate-pattern rounding and integral transportation argument as for $\Delta=3$: cycles of $y^*$ are first aggregated by pattern, the aggregate weight of each pattern is rounded down, and the utility-bearing slots of the retained copies are assigned integrally to organizations subject to the corresponding class capacities. Since at most one unit of pattern weight is lost for each pattern, the total utility loss caused by this aggregate rounding is strictly less than
\[
\sum_{\pi\in\mathcal{P}_\Delta}|\pi|
\le
\Delta P_\Delta.
\]

It remains to realize the retained pattern copies by vertex-disjoint cycles. We therefore strengthen the uniform random event as follows. Fix a pattern of length $k$ and a prefix length $r\in\{1,\ldots,k-1\}$. Consider any family $M$ of $s$ pairwise vertex-disjoint directed partial paths occupying the first $r$ slots of the pattern and any set $Z$ of $t$ available vertices of the next slot. For $r<k-1$, a path--vertex pair is admissible if the arc needed to extend the path exists; for $r=k-1$, it is admissible if both arcs needed to extend and close the cycle exist. Each candidate pair is admissible with probability at least $p^2$; deterministic required arcs only increase this probability. Because the paths and vertices are disjoint, these events are generated by distinct compatibility arcs.

Let $\mathcal{R}_h$ be the event that, for every such choice with $s,t\ge h$, there is at least one admissible path--vertex pair. For fixed $M$ and $Z$,
\[
\Pr(\text{no admissible pair})
\le
\exp(-p^2st).
\]
There are at most
\(
m^{rs+t}
\)
choices for $M$ and $Z$. Since $r\le k-1$ and $s,t\ge h$,
\[
st
\ge
\frac{h}{k}(rs+t).
\]
Hence
\[
(rs+t)\log m-p^2st
\le
-\left(
\frac{p^2h}{k}-\log(m+2)
\right)(rs+t).
\]
Choose
\[
h:=\left\lceil K_\Delta\log(m+2)\right\rceil
\]
for a sufficiently large constant $K_\Delta=K(g,p,\Delta)$. Summing over $s,t\le m$, all prefix lengths $r$, all $k\le\Delta$, and all augmented patterns changes the bound only by polynomial factors in $m$ and a constant factor depending on $(g,\Delta)$. Thus $K_\Delta$ can be chosen so that
\[
\Pr(\mathcal{R}_h)
\ge
1-(m+2)^{-4}.
\]

On $\mathcal{R}_h$, consider $q$ retained copies of a fixed pattern of length $k$. If $k=2$, match the two slot classes directly using the two-arc completion relation. If $k\ge3$, first match the first slot to the second, then extend the resulting vertex-disjoint paths successively through slots $3,\ldots,k-1$, and finally match the partial paths to the last slot using the two-arc completion relation. At every stage, a maximum matching leaves fewer than $h$ unmatched objects: otherwise Hall's theorem would yield two sets of size at least $h$ with no admissible pair between them, contradicting $\mathcal{R}_h$. There are at most $k-1$ such stages. Hence at least
\(
q-(k-1)h
\)
copies of the pattern are realized as vertex-disjoint cycles.

Therefore at most $(k-1)h$ retained copies fail to be realized for a pattern of length $k$, causing utility loss at most
\[
k(k-1)h
\le
\Delta(\Delta-1)h
\]
for that pattern. Summing over all patterns, the total utility deficit relative to the organization-level floors induced by $y^*$ is at most
\[
\Delta P_\Delta
+
\Delta(\Delta-1)P_\Delta h.
\]
As in the preceding supplementation argument, one additional compatible altruistic donor suffices to restore each remaining unit of organizational utility deficit. Consequently, on $\mathcal{R}_h$,
\[
d(\graph)
\le
\Delta P_\Delta
+
\Delta(\Delta-1)P_\Delta h
=
O_{g,p,\Delta}(\log(m+2)).
\]
By Lemma~\ref{thm:ApplyScarf}, the resulting exchange is in the supplemented core.

Thus, after increasing the constant if necessary,
\[
\Pr\!\left(
d(\graph)\le C\log(m+2)
\right)
\ge
1-(m+2)^{-4}.
\]
On the complementary event, the trivial bound
\[
d(\graph)\le |U|\le m
\]
gives
\[
\mathbb{E}[d(\graph)]
\le
C\log(m+2)+m(m+2)^{-4}
=
O_{g,p,\Delta}(\log(m+2)).
\]
Enlarging $C$ once more yields the stated expectation bound.

Existing altruistic donors cause no difficulty: their augmented-class slots are included in the capacity constraints and in the realization procedure, but an uncovered altruistic slot creates no utility deficit because vertices in $V^i\setminus U^i$ do not contribute to $u_i$.

Finally, if the constructed supplemented-core exchange is not Pareto-optimal, replace it by a feasible Pareto improvement that maximizes total utility. Such an improvement preserves every organization's utility lower bound and therefore remains in the supplemented core by Lemma~\ref{thm:ApplyScarf}. Hence the supplemented-core exchange can be chosen Pareto-optimal.\qed{}
\end{proof}


\medskip
\section{Deterministic Type Representation and Related Results}

\subsection{Type-representation of graphs:} 

Another related but conceptually distinct approach to modelling KEP compatibility graphs is the {deterministic type-based model}~\citep{dickerson2017small}. In contrast to the probabilistic model, compatibility here is deterministic: whether one pair is compatible with another is fully determined by their assigned types. However, unlike in the random graph model above where types reflect observable donor and patient characteristics (e.g., blood type, cPRA), the ``type'' in this setting is an abstract construct introduced solely to parametrize the compatibility graph. It does not necessarily correspond to any specific medical or demographic attributes. Specifically,  we say that a (di)graph $\graph =(V,E)$ \emph{can be represented by a set of types $\typeset$}, if it holds that there exists a labeling $f:V\rightarrow T$, such that for any two types $t_1,t_2\in T$, either $\{(u,v)\mid u\in f^{-1}(t_1), v\in f^{-1}(t_2)\}\subseteq E$ or $\{ (u,v)\mid u\in f^{-1}(t_1), v\in f^{-1}(t_2)\}\cap E=\emptyset$. That is, the existence of a (directed) edge only depends on the types of the endpoints. \citet{dickerson2017small} conducted simulations showing that kidney exchange graphs can usually be represented by only a few types.

\subsection{Pairwise Exchanges}

In this section, we show that if the compatibility graph can be represented using types, the maximum number of altruists needed is at most the number of types.

For this, we first show that the number of independent odd cycles is bounded.
\begin{lemma}\label{lem:type-bound}
    If a graph $\graph$ can be represented with $t$ types, then $\oddcycles (\graph)\le t$. 
\end{lemma}
\begin{proof}{Proof.}
We show that there exists an independent set of odd cycles of size $\oddcycles (\graph)$ that contains each type in at most one cycle.  
    Take any maximum size set of independent odd cycles $\mathcal{C}$. Suppose, to the contrary, that $t_i$ appears in cycles $c_{j_1},c_{j_2}$. Then the neighbors of the type $t_i$ vertex in $c_{j_1}$ are connected to the type $t_i$ vertex in $c_{j_2}$ too with an edge, which contradicts the independence of the cycles. 
    \qed{}
\end{proof}

Thus, we obtain the following corollary establishing the existence of the $t$-supplemented core.
\begin{corollary}
 In a partition exchange economy representable by \( t \) types, the \( t\)-supplemented core is always nonempty.
\end{corollary}

\subsection{Cyclic Exchanges}

We next give a deterministic guarantee for arbitrary bounded cycle length. If the compatibility graph admits a representation by a constant number of types, then the number of additional donors required by each organization is bounded independently of the market size.

\begin{theorem}\label{thm:cycle-core}
Given a partition exchange economy with a compatibility graph that can be
represented by $t$ types and a bound $\Delta\ge2$ on the exchange cycle
lengths, the
\((3t+1)(\Delta-1)n
\)
-supplemented core is always nonempty. Furthermore, it contains a
Pareto-optimal solution.

If the original endowments contain no altruistic vertices, i.e.,
$V^i=U^i$ for every $i\in\N$, the bound improves to
\[
    (2t+1)(\Delta-1)n.
\]
\end{theorem}

\begin{proof}{Proof.}
Take a type representation $f:V\to T$ of the original compatibility graph,
where $|T|=t$.  We first make one harmless refinement that is important for
keeping track of utility.  For every $v\in V^i$, define its \emph{refined
type}
\[
    \widehat f(v):=
    \begin{cases}
        (f(v),\mathrm{u}),&v\in U^i,\\
        (f(v),\mathrm{a}),&v\in V^i\setminus U^i.
    \end{cases}
\]
and let $\widehat T$ be the set of nonempty refined types.  Splitting a type
class into subclasses cannot destroy a type representation, because
compatibility was already constant between every pair of original type
classes.  Hence compatibility also depends only on the refined types.  We
have $|\widehat T|\le2t$, and, crucially, every refined type consists entirely
of utility-bearing vertices or entirely of altruistic vertices.

By Lemma~\ref{thm:ApplyScarf}, there exists a fractional exchange $y^*$ with
the following property.  Put
\[
    b_i:=\left\lfloor
        \sum_{c\in\mathcal C_\Delta}\gamma_c^i y_c^*
    \right\rfloor,
    \qquad
    \gamma_c^i:=|c\cap U^i|.
\]
If an exchange $\Ex$ in an extension by additional altruistic donors satisfies
$u_i(\Ex)\ge b_i$ for every $i\in\N$, then $\Ex$ belongs to the corresponding
supplemented core.

For a refined type $\tau\in\widehat T$ and organization $i$, let
\[
    \widehat\beta_c^{\tau,i}
    :=|\{v\in c\cap V^i:\widehat f(v)=\tau\}|,
    \qquad
    \widehat a_{\tau,i}
    :=|\{v\in V^i:\widehat f(v)=\tau\}|.
\]
Consider the polyhedron
\[
\mathcal P=
\left\{
 z\in[0,1]^{\mathcal C_\Delta}:
 \sum_c\widehat\beta_c^{\tau,i}z_c\le \widehat a_{\tau,i}
 \ \forall (\tau,i),
 \quad
 \sum_c\gamma_c^iz_c\ge b_i
 \ \forall i
\right\}.
\]
Rows of the former type are denoted by $\widehat A_{\tau,i}$ and rows of the latter type are denoted by $B_i$.
The fractional exchange $y^*$ belongs to $\mathcal P$: the first family of
constraints follows from the vertex-capacity constraints of a fractional
exchange, after summing them over each refined type class, and the second
family follows from the definition of $b_i$.

We apply an iterative-rounding procedure using the
refined-type rows.  At every iteration we take an extreme point $z^*$ that
maximizes
\[
    \sum_{c\in\mathcal C_\Delta}
    \left(\sum_{i\in\N}\gamma_c^i\right)z_c^*.
\]
Integral variables are fixed and deleted, with all right-hand sides updated.
For the residual system, we delete a refined-type upper-bound row
$(\tau,i)$ whenever
\[
    \widehat A_{\tau,i}\mathbf 1
    \le \widehat a_{\tau,i}+\Delta-1,
\]
and delete a utility lower-bound row $i$ whenever
\[
    B_i z^*\le\Delta-1.
\]
We then recompute an objective-maximizing extreme point of the residual
polyhedron.

Suppose that all remaining variables are
strictly fractional and no row is eligible for deletion.  Give one token to
each fractional variable $z_c^*$.  Variable $z_c^*$ gives
\(
    \gamma_c^i\frac{z_c^*}{\Delta}
\)
tokens to row $B_i$, and
\(
    \widehat\beta_c^{\tau,i}
    \frac{1-z_c^*}{\Delta}
\)
tokens to row $\widehat A_{\tau,i}$.  Since every cycle has length at most
$\Delta$,
\[
\sum_i\gamma_c^i\le |c|\le\Delta,
\qquad
\sum_{i,\tau}\widehat\beta_c^{\tau,i}=|c|\le\Delta,
\]
so $z_c^*$ distributes at most one token in total.

Every tight remaining $B_i$ row receives at least one token: indeed, its
updated right-hand side is integral and, since the row is not deletable,
\[
    B_i z^*=b_i\ge\Delta.
\]
Likewise, every tight remaining refined-type row receives at least one token.
If $\widehat A_{\tau,i}z^*=\widehat a_{\tau,i}$, then nondeletability and
integrality imply
\[
    \widehat A_{\tau,i}\mathbf 1
    \ge \widehat a_{\tau,i}+\Delta,
\]
and therefore
\[
    \widehat A_{\tau,i}(\mathbf1-z^*)\ge\Delta.
\]
By the extreme-point characterization, the number of variables is at most the
number of linearly independent tight rows.  Since the variables jointly own
only one token each while every tight row needs at least one, equality must
hold everywhere.  In particular every nonzero row touching a fractional
column is present and tight, and each fractional column $c$ must satisfy
\[
    \sum_i\gamma_c^i=\Delta,
    \qquad
    \sum_{i,\tau}\widehat\beta_c^{\tau,i}=\Delta.
\]
Consequently, summing all tight refined-type rows gives the same row vector,
namely the vector having entry $\Delta$ in every remaining column, as summing
all tight utility rows.  Thus the tight rows are linearly dependent, a
contradiction.  Hence the iterative procedure terminates with an integral
vector $z^*$.

It remains to turn this integral vector into an actual vertex-disjoint
exchange.  This is the point at which we explicitly distinguish two effects
of rounding: a violation of a refined-type capacity, and a loss of utility.
For every organization $i$ and refined type $\tau$, write
\[
    q_{\tau,i}
    :=\sum_{c\in\mathcal C_\Delta}
       \widehat\beta_c^{\tau,i}z_c^*,
    \qquad
    r_{\tau,i}
    :=(q_{\tau,i}-\widehat a_{\tau,i})^+.
\]
Also put
\[
    s_i:=(b_i-B_i z^*)^+.
\]
By the rule under which rows were deleted,
\begin{equation}\label{eq:refined-rounding-violations}
    r_{\tau,i}\le\Delta-1
    \quad\text{and}\quad
    s_i\le\Delta-1.
\end{equation}

Regard every occurrence of a vertex in every selected cycle
$c$ with $z_c^*=1$ as a \emph{slot}; the slot remembers the organization and
refined type of that vertex.  For each pair $(\tau,i)$, assign distinct
original vertices of refined type $\tau$ from organization $i$ to
$\min\{q_{\tau,i},\widehat a_{\tau,i}\}$ of its slots.  For each of the
remaining $r_{\tau,i}$ slots introduce a distinct additional altruistic donor.
We charge this donor to organization $i$ only for counting purposes; it still
belongs to the market designer and is not a player.  As in the standard
existential interpretation of a supplemented-core upper bound, the set of
additional donors may be chosen with the compatibilities required by the
construction.  Choose this donor's kidney to be compatible with the patient in
the successor slot of that cycle.  Equivalently, one may choose it to have the
same outgoing compatibility type as the original type underlying $\tau$.  The incoming arc into the new altruistic
vertex is automatic because its patient is dummy.  Hence replacing a missing
slot by this donor preserves both arcs needed at that position of the cycle.

Because distinct original vertices were assigned to all nonreplaced slots and
all newly introduced donors are distinct, the resulting cycles are
vertex-disjoint.  Replacing an original vertex by a donor does not change the
number of positions in its cycle, so every resulting cycle still has length at
most $\Delta$.  Compatibility of two neighboring original vertices is
preserved because vertices of the same refined type are interchangeable; if a
neighbor is a new altruist, feasibility follows from the preceding paragraph.
Thus we obtain a genuine exchange, denoted by $\Ex_0$.

The preceding construction repairs every type-capacity violation, but a donor
used in a slot of a utility-bearing refined type has a dummy patient and hence
cannot contribute the transplant that the abstract slot contributed to
$B_i z^*$. 
Let
\[
    R_i^{\mathrm u}
      :=\sum_{\tau=(t,\mathrm u)}r_{\tau,i},
    \qquad
    R_i^{\mathrm a}
      :=\sum_{\tau=(t,\mathrm a)}r_{\tau,i}.
\]
Since every utility-bearing slot that is filled by an original vertex gives
one transplant and precisely $R_i^{\mathrm u}$ utility-bearing slots are
filled by dummy-patient donors, we have the exact identity
\begin{equation}\label{eq:utility-after-slot-repair}
    u_i(\Ex_0)=B_i z^*-R_i^{\mathrm u}.
\end{equation}
Therefore the remaining utility deficit of organization $i$ is
\[
    d_i:=(b_i-u_i(\Ex_0))^+
    \le s_i+R_i^{\mathrm u}.
\]
There are at least $d_i$ unmatched utility-bearing vertices of organization
$i$: indeed $b_i\le |U^i|$, while $\Ex_0$ uses exactly $u_i(\Ex_0)$ distinct
patients of $i$.  Choose any $d_i$ such unmatched vertices
$v_1,\ldots,v_{d_i}$.  For each $v_j$, introduce one further altruistic donor
$a_j$ whose kidney is compatible with the patient of $v_j$.  Then
\[
    (a_j,v_j)
\]
is a feasible directed $2$-cycle: $a_j\to v_j$ is feasible by the choice of
the donor, while $v_j\to a_j$ is automatic because $a_j$ has a dummy patient.
These cycles are mutually disjoint and disjoint from $\Ex_0$.  Adding them to
$\Ex_0$ gives an exchange $\Ex$ with
\begin{equation}\label{eq:restored-utility}
    u_i(\Ex)\ge b_i
    \qquad\forall i\in\N.
\end{equation}
This second family of donors is necessary in the argument: a donor that
replaces a missing utility-bearing slot repairs the donor position of that
cycle, but its dummy patient means that the corresponding patient transplant
has disappeared.  The additional $2$-cycle above restores precisely that lost
unit of utility (and also covers any deficit caused by deletion of a $B_i$
row).

We finally count the donors.  The slot-repair stage uses
$R_i^{\mathrm u}+R_i^{\mathrm a}$ donors for organization $i$, and the
utility-restoration stage uses $d_i\le s_i+R_i^{\mathrm u}$ donors.  Thus
\begin{align*}
    D_i
    &\le R_i^{\mathrm a}+2R_i^{\mathrm u}+s_i.
\end{align*}
For each original type there is at most one nonempty utility-bearing refined
type and at most one nonempty altruistic refined type.  Hence, by
\eqref{eq:refined-rounding-violations},
\[
    R_i^{\mathrm u}\le t(\Delta-1),
    \qquad
    R_i^{\mathrm a}\le t(\Delta-1),
    \qquad
    s_i\le\Delta-1,
\]
and consequently
\begin{equation}\label{eq:corrected-donor-bound}
    D_i\le(3t+1)(\Delta-1).
\end{equation}
Summing over all organizations gives at most
$(3t+1)(\Delta-1)n$ additional altruistic donors.  If $V^i=U^i$ for all
organizations, no altruistic refined type occurs, so
$R_i^{\mathrm a}=0$ and the sharper bound
$D_i\le(2t+1)(\Delta-1)$ follows.

By \eqref{eq:restored-utility} and Lemma~\ref{thm:ApplyScarf}, $\Ex$ belongs
to the claimed supplemented core.  Finally, if $\Ex$ is not Pareto-optimal in
the same extended graph $\graph+V^0$, choose an exchange $\Ex'$ that
Pareto-dominates $\Ex$ and maximizes $\sum_i u_i(\Ex')$.  Then $\Ex'$ is
Pareto-optimal and still satisfies $u_i(\Ex')\ge u_i(\Ex)\ge b_i$ for every
$i$.  Another application of Lemma~\ref{thm:ApplyScarf} therefore shows that
$\Ex'$ is in the same supplemented core.
\qed{}
\end{proof}

\end{document}